\documentclass[12pt, %a4 paper
]{article}
\usepackage{booktabs}

\usepackage[utf8]{inputenc}
\usepackage{amssymb}
\usepackage{amsmath}

\usepackage{tikz}
\usepackage{geometry}
\newtheorem{theorem}{Theorem}%[chapter]

\newtheorem{corollary}{Corollary}%[chapter]
\newtheorem{definition}{Definition}%[chapter]
\newtheorem{example}{Example}%[chapter]
\newtheorem{lemma}{Lemma}%[chapter]
\newtheorem{proposition}{Proposition}%[chapter]
\newtheorem{remark}{Remark}%[chapter]
\newenvironment{proof}[1][Proof]{\emph{#1.} }{\  \hfill $\square $ \vspace{5 pt}}

\usepackage{natbib}
\usepackage{amssymb}
\usepackage{amsmath}

\usepackage{mathpazo}
\usepackage{enumerate}

\usepackage{amsfonts}

\usepackage{amssymb}

\usepackage{enumerate}

\usepackage{mathrsfs}

\usepackage{cancel}

\usepackage{pgf,tikz}

\usepackage{graphics}
\usepackage{color}

\usetikzlibrary{decorations.pathreplacing,angles,quotes}
\usetikzlibrary{arrows.meta}
\usetikzlibrary{arrows, decorations.markings}
\tikzset{myptr/.style={decoration={markings,mark=at position 1 with %
       {\arrow[scale=2,>=stealth]{>}}},postaction={decorate}}}
\usepackage{hyperref}
\hypersetup{
    colorlinks,
    citecolor=blue,
    filecolor=black,
    linkcolor=blue,
    urlcolor=blue
}
\usepackage{pgf,tikz}
\usetikzlibrary{arrows}

\usepackage{fancyhdr}
\usepackage{soul}

\usepackage{emptypage}
\newcommand*\samethanks[1][\value{footnote}]{\footnotemark[#1]}

\usepackage[T1]{fontenc}
\DeclareFontFamily{T1}{calligra}{}
\DeclareFontShape{T1}{calligra}{m}{n}{<->s*[1.44]callig15}{}
\DeclareMathAlphabet\mathcalligra   {T1}{calligra} {m} {n}

\usepackage{multirow}
\usepackage{array}
\usepackage{mathtools}
\usepackage{arydshln}
\usepackage{threeparttable}
\usepackage{booktabs,caption}

\usepackage{ifthen}
\newboolean{showcomments}
\setboolean{showcomments}{true}

\newcommand{\pablo}[1]{  \ifthenelse{\boolean{showcomments}}
{\textcolor{green!50!black}{(T: #1)}}{}}
\newcommand{\marcelo}[1]{\ifthenelse{\boolean{showcomments}}
{\textcolor{red}{(M: #1)}}{}}
\newcommand{\agustin}[1]{  \ifthenelse{\boolean{showcomments}}
{\textcolor{blue!50!black}{(T: #1)}}{}}

\begin{document}

\title{Not obviously manipulable allotment rules%
\thanks{%
%Thanks to be added. 
We thank Alejandro Neme, James Schummer, Fernando Tohmé, and participants of the 23rd annual SAET Conference for their helpful comments. We acknowledge the financial support
from UNSL through grants 032016, 030120, and 030320, from Consejo Nacional
de Investigaciones Cient\'{\i}ficas y T\'{e}cnicas (CONICET) through grant
PIP 112-200801-00655, and from Agencia Nacional de Promoción Cient\'ifica y Tecnológica through grant PICT 2017-2355.}}


\author{R. Pablo Arribillaga\thanks{%{\scriptsize 
Instituto de Matem\'{a}tica Aplicada San Luis (UNSL-CONICET) and Departamento de Matemática, Universidad Nacional de San
Luis, San Luis, Argentina. Emails: \href{mailto:rarribi@unsl.edu.ar}{rarribi@unsl.edu.ar} (R. P. Arribillaga) %(A. G. Bonifacio) 
and \href{mailto:abonifacio@unsl.edu.ar}{abonifacio@unsl.edu.ar} (A. G. Bonifacio).
%}
} \and Agustín G. Bonifacio\samethanks[2] 
}

\date{\today}

\maketitle

\begin{abstract}

%In the problem of allocating a single non-disposable commodity among agents whose preferences are single-peaked, we study the  manipulation of own-peak-only rules. We characterize a large subfamily of these rules as the only ones that are efficient, satisfy a minimal fairness requirement, and are  not obviously manipulable.  We also identify the set of single-plateaued preferences as a maximal domain of preferences, including the set of single-peaked preferences, for those properties.    

%\vspace{20 pt}

%In the problem of allocating a single non-disposable commodity among agents whose preferences are single-peaked, we study a weakening of strategy-proofness called  not obvious manipulability (NOM). Classical results in the literature show that: (i) the uniform rule is the only one that satisfies strategy-proofness, efficiency, and symmetry; and (ii) the single-plateaued domain is maximal for those three properties. If agents are cognitively limited and NOM is sufficient to describe their strategic behavior, two important questions arise: (i) ¿are there other relevant rules to consider?, and (ii) ¿does this maximal domain get larger? To address these questions, in this paper, we: (i) characterize a large and simple class of rules as the only ones that satisfy NOM, efficiency, and a mild fairness requirement; and (ii) show that the single-plateaued domain is still maximal for these new properties.    

In the problem of allocating a single non-disposable commodity among agents whose preferences are single-peaked, we study a weakening of strategy-proofness called not obvious manipulability (NOM). If agents are cognitively limited, then NOM is sufficient to describe their strategic behavior. We characterize a large family of own-peak-only rules that satisfy efficiency, NOM, and a minimal fairness condition. We call these rules "simple". %\textcolor{red}{The idea behind their definition has been dormant in the axiomatic literature for a long time.} \textcolor{blue}{The identification of this full family of rules has been dormant in the axiomatic literature for a long time.}
In economies with excess demand, simple rules fully satiate agents whose peak amount is less than or equal to equal division and assign, to each remaining agent, an amount between equal division and his peak. In economies with excess supply, simple rules are defined symmetrically. These rules can be thought of as a two-step procedure that involves solving a claims problem.
We also show that the single-plateaued domain is maximal for the characterizing properties of simple rules. Therefore, even though replacing strategy-proofness with NOM greatly expands the family of admissible rules, the maximal domain of preferences involved remains basically unaltered.  
%We characterize a large and easy-to-describe class of rules which are NOM, that we call ``simple''.  

%DECIR ALGO DE LAS TRES REGLAS FAMOSAS??

\bigskip

 \bigskip


\end{abstract}

%\newpage
\section{Introduction}

%%%%%%%%%%%%%%%%%%%%%%%%%%%%%%%%%%%%%%%%%%%%%%%%%%%%%%%%%%%%%%%%%%%%%%%%%%%%%%%%%%%%

Consider the problem of allocating a single non-disposable commodity among a group of agents with single-peaked preferences: up to some critical level, called the peak, an increase in an agent's consumption raises his welfare; beyond that level, the opposite holds. In this context, an allotment rule is a systematic procedure that allows agents to select an allotment, among
many, according to their preferences.  
% Within desirable properties allotment rules may satisfy, the concept of strategy-proofness has played a central role in studying the  strategic behavior of the agents. 
An allotment rule is strategy-proof if misreporting preferences is never better than truth-telling. In this paper, we study a weakening of strategy-proofness called not obvious manipulability.

Criteria of incentive compatibility, efficiency, and fairness are considered the three perennial goals of mechanism design. When those three criteria are materialized in the properties of strategy-proofness, efficiency, and symmetry, they 
%, .  %stand out: efficiency, strategy-proofness, and fairness. 
%The combination of these three goals 
single out the ``uniform'' rule \citep{sprumont1991division, ching1994alternative}; whereas relaxing symmetry by permitting an asymmetric treatment of agents expands the options to the versatile family of ``sequential'' allotment rules \citep{barbera1997strategy}.  Insisting on efficiency and on a minimal fairness condition, and in line with a current strand in the literature, we investigate what happens when strategy-proofness is weakened. The idea underlying this approach is that even though manipulations are pervasive,  agents may not realize they can manipulate a rule because they lack information about others' behavior or they are cognitively limited. 
%We will consider an agent that, for each of his reports, knows the possible resulting outcomes of the rule but is cognitively limited in the sense that he is unable to know the possible outcomes \emph{contingent on the preferences of the other agents.} A straightforward rationale for this cognitive limitation is that the agent has limited information on other agents' preferences. 
A misreport that in a specific situation (a specific preference profile of the other agents) is a profitable manipulation may actually be worse than truth-telling in another situation. Therefore, if we consider an agent with limited information on other agents' preferences,  it may be very unclear whether such manipulation will be profitable in practice. Assuming that the agent knows all possible outcomes under any misreport and under truth-telling, there are several different ways to determine when a misreport is recognized or identified as a profitable manipulation. \cite{troyan2020obvious}, in the context of two-sided matching, propose a simple and tractable way to do this by considering best/worst-case scenarios to formalize the notion of obvious manipulation. A manipulation is obvious if it either makes the agent better off than truth-telling in the worst case or makes the agent better off than truth-telling in the best case. An allotment rule is not obviously manipulable (NOM) if it has no obvious manipulation.

 If agents are cognitively limited, then NOM is sufficient to describe their strategic behavior. Therefore, the question arises to what extent NOM rules enrich the landscape of strategy-proof rules. %In order to keep our analysis tractable, 
 We will focus on own-peak-only rules. This means that the sole information collected by these rules from an agent’s preference to determine his allotment is his peak amount.  Because of their simplicity, own-peak-only rules are important rules in their own right and are both useful in practice and extensively studied in the literature. Furthermore, the own-peak-only property follows from efficiency and strategy-proofness  \citep[see][]{sprumont1991division,ching1994alternative}. However, since we do not impose strategy-proofness,  we explicitly invoke it here.

Various rules have been proposed to address these allocation problems. One rule that stands out as the best-behaved
from several standpoints is the ``uniform rule''  %\citep[][]{sprumont1991division}, 
which seeks to equalize the amounts received by agents while maintaining efficiency. As we already mentioned, this rule is strategy-proof and, therefore, NOM. Additionally, the literature highlights two other important rules: the ``constrained equal-distance rule'' which measures agents' sacrifices based on the distance between assignments and peak amounts and strives to equalize those sacrifices, and the ``proportional rule'' which operates similarly to the former but gauges the sacrifice agents make in proportion to their peak amounts \citep[see][for further details on these rules]{thomson2014fully}. However, both the constrained equal-distance and proportional rules are incompatible with NOM.

Our goal is to introduce and characterize a large family of own-peak-only rules satisfying NOM, which we call "simple".\footnote{Within this class, we can construct rules that partially recover constrained equal-distance and proportional principles while remaining compatible with NOM.}
%The idea behind their definition has been dormant in the axiomatic literature for a long time. 
Their definition is as follows. In economies with excess demand, simple rules fully satiate agents whose peak amount is less than or equal to equal division and assign an amount between equal division and his peak to each remaining agent. Symmetrically, in economies with excess supply, simple rules fully satiate agents whose peak is greater than or equal to equal division and assign to each remaining agent an amount between his peak and equal division. 
Agents that are fully satiated by these rules, which we also call ``simple'', are thus rewarded since they are responsible for all exchanges. Simple rules can also be thought of as a two-step procedure. In the first step, simple agents are fully satiated and the rest are provisionally assigned equal division. In the second step, by solving a claims problem, the provisional assignments are adjusted to attain overall feasibility, provided that peak amounts are not overpassed.

We believe that the full family of simple rules has been dormant in the literature because many of the rules (and families of rules) identified in the literature belong to this large family. This is the case, for example, of the family of sequential rules in \cite{barbera1997strategy} that are defined throughout an intricated process needed to cope with strategy-proofness.

In our main result, Theorem \ref{characterization}, we show that an allotment rule is simple if and only if it is own-peak-only, efficient, satisfies NOM, and meets the equal division guarantee. This last property is a minimal fairness requirement that states that whenever an agent demands equal division, the rule must guarantee him that amount. If we narrow the picture to symmetric rules, i.e., rules that assign indifferent allotments to agents with the same preferences, then we also obtain a characterization of a big subfamily of simple rules that satisfies all five properties (Corollary \ref{characterization 2}).\footnote{Note that although symmetry and the equal division guarantee can both be considered as fairness requirements, they are logically independent properties.} As we previously said, in contrast, when we replace NOM with strategy-proofness the only efficient and symmetric rule that remains is the uniform rule \citep{ching1994alternative}.   
We also provide two variants of our main characterization: %(i) changing own-peak-onliness by the more usual (and stronger) peaks-only requirement (Theorem \ref{characterization peaks-only}), 
(i) adapting the result to economies with individual endowments (Proposition \ref{characterization bis}) and (ii) invoking a peak responsiveness property that also encompasses symmetry and the equal division guarantee (Proposition \ref{characterization3}). %, and (iii) getting rid of efficiency (Proposition  \ref{characterization generalized simple}).  %These results show that departing from strategy-proofness to the less demanding NOM property expands greatly the admissible family of allotment rules   

Next, we analyze the maximality of the domain of preferences (including the domain of single-peaked preferences) for which a rule satisfying  own-peak-onliness, efficiency, the equal division guarantee, and NOM  exists. For the properties of efficiency, strategy-proofness, and symmetry, the single-plateaued domain is maximal \citep{ching1998maximal,masso2001maximal}. In Theorem \ref{theo max domain}, we show that the single-plateaued domain is maximal for our properties as well.  Therefore, even though replacing strategy-proofness with NOM greatly expands the family of admissible rules, the maximal domain of preferences involved remains basically unaltered.

To the best of our knowledge, our paper is the first one that applies \cite{troyan2020obvious} notion of obvious manipulations to the allocation of a non-disposable commodity among agents with single-peaked preferences. In the context of voting, \cite{aziz2021obvious} and  \cite{arribillaga2024obvious} study non-peaks-only and peaks-only rules, respectively.  Other recent papers that study obvious manipulations, in other situations, are  \cite{ortega2022obvious}, \cite{psomas2022fair}, and \cite{arribillaga2023obvious}.

%%%%%%%%%%%%%%%%%%%%%%%%%%%%%%%%%%%%%%%%%%%%%%%%%%%%%%%%%%%%%%%%%%%%%%%%%%%%%%%%%%%%

The rest of the paper is organized as follows. The model and the concept of obvious manipulations are introduced in Section \ref{Preliminaries}. In Section \ref{simple section}, after analyzing three classical rules, we present simple rules and their main characterization. Section \ref{section further} introduces some related characterizations. The maximal domain result is presented in Section \ref{section maximal}. To conclude, some final remarks are gathered in Section \ref{section final}. 

\section{Preliminaries}\label{Preliminaries}

\subsection{Model} \label{2.1}

A social endowment $\Omega \in \mathbb{R}_{++}$ is an amount of a perfectly divisible commodity to be distributed among a set of agents $N=\{1,2,\ldots, n\}$. Each $i \in N$ is equipped with  a  continuous  preference relation $R_i$ defined over $\mathbb{R}_+ \cup \{\infty\}$. 
%with a unique peak, $p(R_i)$.
Call $P_i$ and $I_i$ to the strict preference and indifference relations associated with $R_i,$ respectively. Denote by $\mathcal{U}$ the domain of all such preferences. Given  $R_i\in \mathcal{U}$, let $p(R_i)=\{x\in\mathbb{R}_+ \cup \{\infty\}: xR_iy \text{ for each } y\in\mathbb{R}_+\cup \{\infty\} \}$ be the set of preferred consumptions
according to $R_i$, called the \textbf{peak} of $R_i$. When $p(R_i)$ is a singleton, we slightly abuse notation and use $p(R_i)$ to denote its single element. Agents $i$'s preference $R_i \in \mathcal{U}$ is \textbf{single-peaked} if $p(R_i)$ is a singleton and, for each pair $\{x_i, x_i'\} \subseteq \mathbb{R}_+$, we have $x_iP_ix_i'$ as long as either $x_i'<x_i\leq p(R_i)$ or $p(R_i) \leq x_i<x_i'$ holds. Denote by $\mathcal{SP}$ the domain of all such preferences.

A (generic) domain of preferences $\mathcal{D}$ is a subset of $\mathcal{U}.$ Given  a domain of preferences $\mathcal{D}\subseteq \mathcal{U}$, an  \textbf{economy} in $\mathcal{D}$ 
consists of a profile of preferences $R=(R_j)_{j \in N} \in \mathcal{D}^n$ and a social  endowment $\Omega \in \mathbb{R}_{++}$ and is denoted by $(R,\Omega)$. Let  $\mathcal{E}_{\mathcal{D}}$ be the domain of all such economies. Given a social endowment $\Omega \in \mathbb{R}_{++}$, the set of \textbf{(feasible) allotments} of $\Omega$ is $X(\Omega)=\{x \in  \mathbb{R}^n_+ : \sum_{j\in N}x_j= \Omega\}$.
%  be the set of \textbf{allocations} for economy $(R, \Omega)$,  and 
%Let $X=\bigcup_{(R, \Omega) \in \mathcal{E}_{\mathcal{D}}}X(R, \Omega).$  
An   \textbf{(allotment) rule} on $\mathcal{E}_\mathcal{D}$ is a function $\varphi: \mathcal{E}_\mathcal{D} \longrightarrow \mathbb{R}^n_+$ such that $\varphi(R, \Omega) \in X(\Omega)$ for each $(R, \Omega) \in \mathcal{E}_\mathcal{D}.$ 
%Let $\mathcal{SPE}$ denote the set of all economics $(R,\Omega)\mathcal{E}$ such that $R \in \mathcal{SP}$
%For each $N \in \mathcal{N},$ each $i\in N,$ and each $(R, \Omega) \in \mathcal{E}_{\mathcal{SP}},$ let $\Delta \varphi_i(e)=\varphi_i(e)-\omega_i$ be \textbf{agent $\boldsymbol{i$'s net trade at $e}$}.   

Given a rule  $\varphi$ defined on a generic domain $\mathcal{E}_\mathcal{D}$, some desirable properties we consider are listed next. 

\vspace{5 pt}

\noindent \textbf{Efficiency:} For each $(R, \Omega) \in \mathcal{E}_\mathcal{D}$, there is no $x \in X(\Omega)$ such that $x_iR_i\varphi_i(R, \Omega)$ for each $i \in N$ and $x_iP_i\varphi_i(R, \Omega)$ for some $i \in N.$ 

\vspace{5 pt}

\noindent Efficiency is the usual Pareto optimality criterion. Under this condition, for each economy, the allocation selected by the rule should be such that there is no other allocation that all agents find at least as desirable and at least one agent (strictly) prefers. As usual, throughout the paper, we assume that all rule satisfies this property.

Next, we introduce a useful property.% that is equivalent to efficiency on $\mathcal{E}_\mathcal{SP}.$ 

\vspace{5 pt}
\noindent
\textbf{Same-sidedness:} For each $(R,\Omega) \in \mathcal{E}_\mathcal{SP}$,
\begin{enumerate}[(i)]
    \item $\sum_{j \in N}p(R_j) \geq \Omega$ implies  $\varphi_i(R)\leq p(R_i)$ for each $i \in N,$ and
    \item $\sum_{j \in N}p(R_j) \leq \Omega$ implies  $\varphi_i(R)\geq p(R_i)$ for each $i \in N$.
\end{enumerate}
\vspace{5 pt}
%\noindent

%\textbf{Same-sidedness:} For each $(R,\Omega) \in \mathcal{E}_\mathcal{SP}$,  $\sum_{j \in N}p(R_j) \geq \Omega$ implies  $\varphi_i(R)\leq p(R_i)$ for each $i \in N,$ 
%and  $\sum_{j \in N}p(R_j) \leq \Omega$ implies  $\varphi_i(R)\geq p(R_i)$ for each $i \in N.$
%\vspace{5 pt}
\begin{remark}\label{same eff}
Let $\varphi$ be a rule defined on $\mathcal{E}_\mathcal{SP}$. Then, $\varphi$ is efficient if and only if it is same-sided. 
\end{remark}

A rule is strategy-proof if, for each agent,truth-telling is always optimal, regardless of the preferences declared by the other agents. To define it formally, let $R_i, R_i' \in \mathcal{D}$, and $\Omega \in \mathbb{R}_{++}.$ Preference   $R_{i}^{\prime }$ is a \textbf{manipulation  of  
$\boldsymbol{\varphi$ at $(R_i,\Omega)}$} if there is $R_{-i} \in \mathcal{D}^{n-1}$ such that $\varphi_i(R_i', R_{-i}, \Omega)P_i\varphi_i(R_i, R_{-i}, \Omega).$

\vspace{5 pt}

\noindent
\textbf{Strategy-proofness:} For each $i \in N$ and each  $(R_i, \Omega) \in \mathcal{D} \times \mathbb{R}_{++},$ there is no manipulation of $\varphi$ at $(R_i, \Omega)$.

\vspace{5 pt}

The following is an informational simplicity property stating that if an agent unilaterally changes his preference for another one with the same peak, then his allotment remains unchanged.\footnote{This property is weaker than the ``peak-only'' property, that has been imposed in a number of axiomatic studies. See Section \ref{section final} for more details.} 

%\textcolor{blue}{Una propiedad que se deduce de eff y SP es own-peak-only. Como vamos a trabajar con un debilitamiento de SP será necesario pedirla explicitamente.} 

\vspace{5 pt}

\noindent
\textbf{Own-peak-onliness:} For each  $(R, \Omega) \in \mathcal{E}_\mathcal{D},$ each $i \in N,$ and each $R_i' \in  \mathcal{D}$ such that $p(R_i')=p(R_i),$ we have $\varphi_i(R, \Omega)=\varphi_i(R_i', R_{-i}, \Omega).$

\vspace{5 pt}

\noindent Analyzing the uniform rule, \cite{sprumont1991division} derives the own-peak-only property from efficiency and strategy-proofness \citep[see also][]{ching1994alternative}. Since we do not impose strategy-proofness, %in order to keep our analysis tractable 
we explicitly invoke own-peak-onliness.

A well-known fairness property states that agents with the same preferences should be assigned indifferent allocations.

\vspace{5 pt}
\noindent
\textbf{Symmetry:} For each   $(R,\Omega) \in \mathcal{D}$ and each $\{i,j\} \subseteq N$ such that $R_i=R_j$, $\varphi_i(R, \Omega)I_i\varphi_j(R, \Omega)$.

Next, we introduce a new minimal fairness condition that formalizes the idea that each agent has the right to equal division. It states that whenever equal division is claimed, it has to be guaranteed. 

%guarantees equal division when this amount is the most preferred for an agent.  

%It states that whenever an agent has equal division as his preference's peak, the rule must assign it to him.  

\vspace{5 pt}

\noindent
\textbf{Equal division guarantee:} For each   $(R,\Omega) \in \mathcal{D}$ and each $i \in N$ such that $\frac{\Omega}{n}\in p(R_i)$, we have $\varphi_i(R)I_i\frac{\Omega}{n}$.

\vspace{5 pt}

%In this paper, we introduced the equal division guarantee as a minimal fairness requirement. 
\noindent As we will see, the equal division guarantee and symmetry are independent properties. Moreover, in Appendix \ref{on the EDG} we show that two well-known fairness properties, envy-freeness and the equal division lower bound, imply the equal division guarantee for own-peak-only rules.

%\vspace{5 pt}

%\textcolor{red}{NOM:}

%La unica que existe es la uniforme... si relajamos SP a NOM que será presentada a continuacion, manteniendo peaks-only ¿qué otras reglas aparecerian? 

%\textcolor{red}{We consider the set of natural numbers $\mathbb{N}$ as the set of  \textbf{potential agents.} Denote by $\mathcal{N}$ the collection of all finite subsets of   $\mathbb{N}.$ and let $\mathcal{E}_{\mathcal{D}}=\bigcup_{N \in \mathcal{N}} \mathcal{E}_{\mathcal{SP}}_{\mathcal{D}}$}

%\noindent
%\textbf{Consistency:} For each pair $N, N' \in \mathcal{N}$ such that $N' \subseteq N$, each $(R, \Omega) \in \mathcal{E}_{\mathcal{SP}}$, and each $i \in N'$,  we have  $\varphi_i(R_{N'}, \sum_{j \in N'}\varphi_j(R, \Omega))=\varphi_i(R, \Omega)$.

%\vspace{25 pt}

%\noindent \textbf{Uniform rule, $\boldsymbol{u}$:} for each $N \in \mathcal{N},$ each $ (R, \Omega)\in \mathcal{E}_{\mathcal{SP}},$ and each $i \in N,$
%$$u_i(R, \Omega)=\left\{\begin{array}{l l }\min\{p(R_i), \lambda\} & \text{if } z(R, \Omega)\geq 0\\ \max\{p(R_i), \lambda\} & \text{if } z(R,\Omega)< 0\\ \end{array}\right.$$ where $\lambda \geq 0$ and solves $\sum_{j \in N}u_j(R, \Omega)=\Omega.$

%\vspace{25 pt}

%The uniform rule satisfies all the aforementioned properties and \ldots (poner caracterizaciones?) Poner el Lemma 1 de Sonmez...
\subsection{Obvious manipulations}

The notion of obvious manipulation is introduced by \cite{troyan2020obvious}.  %when applied to school choice models and later it has been  studied by \cite{aziz2021obvious} in the context of voting.
They try to single out those manipulations that are easily identifiable by the agents. 
 A manipulation is obvious  if  the best possible outcome under the manipulation is strictly better than the best possible outcome under truth-telling or 
the worst possible outcome under the manipulation is strictly better than the worst possible outcome under truth-telling. %\footnote{\cite{troyan2020obvious} also consider that a manipulation is obvious when the best possible outcome under the manipulation is strictly better than the best possible outcome under truth-telling. 
However, under efficiency, the best possible outcome under truth-telling for an agent is always a peak alternative for that agent in our model. Therefore, no manipulation becomes obvious by considering best possible outcomes and so we omit this in our definition. The derived notion of not obvious manipulation then becomes equivalent to the notion of maximin-strategy-proofness proposed by \cite{brams2006better,brams2008proportional} in the context of cake-cutting.

%  this requiremen is  Poner lo de \begin{equation}\label{cond best}
%B(P_i,O^f(P_{i}^{\prime })) \ P_{i} \ B(P_i, O^f(P_{i})).  
%\end{equation}

 When the set of alternatives is infinite, worst possible outcomes may not be well-defined and a more general definition is necessary in this case. We say that a manipulation is obvious if  each possible outcome under the manipulation is strictly better than some possible outcome under truth-telling. 

%Definition \ref{def OM} is an extension to a model with infinite alternatives of the definition introduced by \cite{troyan2020obvious} in a model with a finite number of alternatives. I

Before we present the formal definition, we introduce some notation. Given a rule $\varphi$ defined on $\mathcal{E}_\mathcal{D}$ and $(R_i, \Omega) \in \mathcal{D} \times \mathbb{R}_{++}$, the \textbf{option set} attainable with $(R_{i}, \Omega)$ at $\varphi$  is 
\begin{equation*}
O^\varphi(R_{i},\Omega)=\left\{x \in [0, \Omega] %\text{ such that }
\ : \ x=\varphi_i(R_{i},R_{-i}) \text{ for some } R_{-i}\in \mathcal{D}^{n-1}\right\}.\footnote{\cite{barbera1990strategy} were the first to use option sets in the context of preference aggregation.}
\end{equation*}
%Let $f:\mathcal{P}^n \longrightarrow X$ be a rule and let $P_i, P_i' \in \mathcal{P}$.  Report $P_i'$ is a  \textit{(profitable) manipulation of rule $f$ at $P_i$} if  there is  a preference sub-profile  $P_{-i} \in \mathcal{P}^{n-1}$ such that 
%\begin{equation*}
%f(P_{i}^{\prime },P_{-i})P_{i}f(P_{i},P_{-i}).
%\end{equation*}
 %Given  $Y \subseteq [0, \Omega],$ denote by %$B(P_i,Y)$ to the best alternative in $Y$ according to preference $P_i$, and by 
 %$W(R_i,Y)$ to the worst alternatives in $Y$ according to preference $R_i$. 

\begin{definition}\label{def OM}
Let $\varphi$  be a rule defined on $\mathcal{E}_\mathcal{D}$ and let $R_i, R_i' \in \mathcal{D}$ and $\Omega \in \mathbb{R}_{++}$ be such that   $R_{i}^{\prime }$ is a manipulation  of  
$\varphi$ at $(R_i,\Omega)$. Then,  $R_i'$ is an  \textbf{obvious manipulation of $\boldsymbol{\varphi$ at $(R_i,\Omega)}$} if for each 
$x' \in O^\varphi(R_{i}^{\prime },\Omega)$ there is $x \in O^\varphi(R_{i},\Omega)$ such that $x'P_i x.$ 
The rule $\varphi$ is \textbf{not obviously manipulable (NOM)} if it does not admit any obvious manipulation. Otherwise, we say that the rule is  \textbf{obviously manipulable (OM)}.
\end{definition}

The next proposition shows that our definition generalizes  \cite{troyan2020obvious}'s definition. In fact, both definitions are equivalent whenever worst possible outcomes exist. Given $R_i$ and $Y\subseteq [0, \Omega]$, denote by $W(R_i, Y)$ the element in $Y$ (if any) such that $xP_i W(R_i, Y)$ for each $x \in Y\setminus \{W(R_i, Y)\}.$

\begin{proposition}\label{equivalencia}
Let $\varphi$ be a rule defined on $\mathcal{E}_\mathcal{D}$ and let $R_i, R_i' \in \mathcal{D}$ and $\Omega \in \mathbb{R}_{++}$ be such that   $R_{i}^{\prime }$ is a manipulation  of  
$\varphi$ at $(R_i,\Omega)$. Assume that $W(R_i,O^\varphi(R_{i}, \Omega))$ and $W(R_i,O^\varphi(R'_{i}, \Omega))$ exist. Then, $R'_{i}$ is an obvious manipulation if and only if $W(R_i,O^\varphi(R_{i}^{\prime },\Omega) \ P_i \  W(R_i, O^\varphi(R_{i},\Omega))$.

\end{proposition}
\begin{proof}
($\Longrightarrow$) Let $R'_{i}$  be an obvious manipulation of $\varphi$ at $(R_i,\Omega)$. As $W(R_i, O^\varphi(R_{i}',\Omega)) \in  O^\varphi(R_{i}',\Omega)$, there is   $x \in O^\varphi(R_{i},\Omega))$ such that $W(R_i, O^\varphi(R_{i}',\Omega))P_ix$. Then, by definition of $W(R_i, O^\varphi(R_{i},\Omega))$, 
\begin{equation}\label{def equivalencia}
    W(R_i, O^\varphi(R_{i}',\Omega)) \ P_i \ W(R_i, O^\varphi(R_{i},\Omega))
\end{equation}
\noindent ($\Longleftarrow$) Assume that \eqref{def equivalencia} holds. We will see that $R_i'$ is an obvious manipulation of $\varphi$ at $(R_i,\Omega)$. Let $x' \in O^\varphi(R_{i}',\Omega).$ Then, by definition of $W(R_i, O^\varphi(R_{i}',\Omega))$ and \eqref{def equivalencia},  $$x'R_iW(R_i, O^\varphi(R_{i}',\Omega)) \ P_i \ W(R_i, O^\varphi(R_{i},\Omega)).$$ As $W(R_i, O^\varphi(R_{i},\Omega)) \in O^\varphi(R_{i},\Omega)$, the proof is complete. 
\end{proof}

\section{NOM rules} \label{simple section}

In this section, we focus on the well-known and widely used domain of single-peaked preferences. In subsection \ref{famous}, three famous rules are asses concerning their obvious manipulability. In subsection \ref{simple} we introduce a new class of rules called ``simple''  that play a relevant role in our characterizations and some examples of relevant simple rules are given.
In subsection \ref{subsection characterizations}, our two main characterizations are presented. In subsection \ref{subsection independence}, the independence of the axioms involved in the characterizations is discussed.

\subsection{Tree famous rules} \label{famous}

In economies with single-peaked preferences, one rule stands out as the best-behaved from several standpoints: the ``uniform rule'' \citep{sprumont1991division}. This rule tries to equate the amounts agents receive as much as possible, maintaining the allocation's efficiency.  Two other important rules in the literature are the ``constrained equal-distance rule'', which uses the distance between assignments and peak amounts as a measure of the sacrifice made by the agents and seeks to equate those sacrifices as much as possible; and the ``proportional rule'', that works similarly to the previous rule but measures the sacrifice agents make in proportion to their peak amounts \citep[see][for further details on these rules]{thomson2014fully}.
%PARA LA constrained equal-distance If we use the distance between assignments and peak amounts as a measure of the sacrifice made by the agents at an allocation, this criterion tries to equate that sacrifice among agents as much as possible 
Their formal definitions are as follows: 
\begin{itemize}
    \item \textbf{Uniform rule, $\boldsymbol{U}$:} for each $(R, \Omega) \in \mathcal{E}_{\mathcal{SP}}$ and each $i \in N$,  
    $$U_i(R, \Omega)=\left\{\begin{array}{l l }
\min\{p(R_i), \lambda\} & \text{if }  \sum_{j \in N} p(R_j) \geq \Omega\\
\max\{p(R_i), \lambda\} & \text{if } \sum_{j \in N} p(R_j) < \Omega\\
\end{array}\right.$$
where $\lambda \geq 0$ and solves $\sum_{j \in N}U_j(R, \Omega)=\Omega.$
\vspace{5 pt}

\item \textbf{Constrained equal-distance rule, $\boldsymbol{CED}$:}  for each $(R, \Omega) \in \mathcal{E}_{\mathcal{SP}}$ and each $i \in N$,  
$$CED_i(R, \Omega)=\left\{\begin{array}{l l }
\max\{p(R_i)-d, 0\} & \text{if } \sum_{j \in N} p(R_j) \geq \Omega\\
p(R_i)+d & \text{if } \sum_{j \in N} p(R_j) < \Omega\\
\end{array}\right.$$
where $d \geq 0$ and solves $\sum_{j \in N}CED_j(R, \Omega)=\Omega.$
\vspace{5 pt}

\item \textbf{Proportional rule, $\boldsymbol{Pro}$:} for each $(R, \Omega) \in \mathcal{E}_{\mathcal{SP}}$ and each $i \in N$,  

$$Pro_i(R, \Omega)=\left\{\begin{array}{l l }
\frac{p(R_i)}{\sum_{j\in N} p(R_j)} \Omega & \text{if } \sum_{j\in N} p(R_j)> 0\\
\frac{\Omega}{n} & \text{if } \sum_{j\in N} p(R_j)= 0\\
\end{array}\right.$$

\end{itemize}

\vspace{10 pt}

It is well known that the uniform rule is strategy-proof on the single peak domain and therefore trivially NOM. As we will see, this rule belongs to a wide and versatile class of NOM rules. The other two discussed rules, though, are OM. We state this fact next.

\begin{proposition} \label{negative} The constrained equal-distance and proportional rules are OM on $\mathcal{E}_{\mathcal{SP}}$.
    
\end{proposition}
\begin{proof}
First, we prove that the constrained equal-distance rule is OM. Let $N=\{1,2\}$ and consider $(R, \Omega)\in \mathcal{E}_{\mathcal{SP}}$ such that $\Omega=1$, $p(R_1)=\frac{1}{3}$ and $0 \ R_1 \ \frac{1}{2}$, and $p(R_2)=0$. Let $R_1' \in \mathcal{SP}$ be such that $p(R'_1)=0$. By definition of the rule and single-peakedness, $CED_1(R_1',R_2,\Omega)=\frac{1}{2} \ P_1 \ \frac{2}{3}=CED_1(R, \Omega)$, so $R'_1$ is a  manipulation of $CED$ at $(R_1, \Omega)$. Furthermore, for an arbitrary $R_2' \in \mathcal{SP}$, $CED_1(R_1',R_2', \Omega) \leq \frac{1}{2}$. Therefore, 
$$W(R_1,O^{CED}(R_{1}^{\prime },\Omega)) \ P_1 \  W(R_1, O^{CED}(R_{1},\Omega))$$
and $R'_1$ is an obvious manipulation of $CED$ at $(R_1, \Omega)$. The proof that the proportional rule is OM is similar and can be carried out using this same example, so we omit it. 
%and proportional rules are OM.
%Assume that $\Omega=1$ and $N=\{1,2\}.$ Assume that $R_1\in \mathcal{SP}$ with $p(R_i)=1/3$ and $0R_i 1/2$ is the true preference of agent $1$ and $R_2\in \mathcal{SP}$ with $p(R_i)=0$ is the preference of agent $2$. Let $R'_1\in \mathcal{SP}$ a preference of agents $1$ such that with $p(R'_1)=0$  
%Let $\varphi$ be the constrained equal-distance or proportional rule. We prove that agent $R'_1$ is an obvious manipulation for agents $1$  of $\varphi$ at $R_1$. By definition of $\varphi$ and single-peakedness, $\varphi_1((R'_1,R_2), \Omega)=1/2 R_1 \varphi_1((R_1,R_2),\Omega)$. Furthermore, for an arbitrary $R'_2\in \mathcal{SP}$,
%$\varphi_1((R'_1,R'_2) \Omega)<1/2$. Therefore, by single-peakedness, definition of $R_1$, and the fact that $=\varphi_1((R_1,R_2)\geq2/3$ 
%$$W(R_1,O^\varphi(R_{1}^{\prime },\Omega)P_1 W(R_1, O^\varphi(R_{1},\Omega))$$
%and $R'_1$ is an obvious manipulation for agents $1$ of $\varphi$ at $R_1.$
%Second, we prove that agent $R'_1$ is an obvious manipulation for agents $1$ of the proportional rule at $R_1$. By definition and single-peakedness, $Prop_1((R'_1,R_2), \Omega)=1/2 R_1 1=Prop_1((R_1,R_2), \Omega)$. Furthermore, for an arbitrary $R'_2\in \mathcal{SP}$,
%$Prop_1((R'_1,R'_2) \Omega)\leq1/2$. Therefore, by single-peakedness and the fact that $=Prop_1((R_1,R_2)=1$ 
%$$W(R_1,O^{Prop}(R_{1}^{\prime },\Omega)P_1 W(R_1, O^{Prop}(R_{1},\Omega))$$
%and $R'_1$ is an obvious manipulation for agents $1$ of the proportional rule at $R_1.$
\end{proof}

In the following subsection, we present a rich and flexible class of rules that, at least partially, allow us to make NOM compatible with  proportionality and the equal-distance criterion, among others.

%POR ACA... EN la siguiente section presentamos una clase de reglas, que llamaremos simples con una amplia riqueza y flexibilidad .... que son NOM. Dentro de esta clase es possible  compatibilizar,al menos parcialente, los criterios de proporcionlaidad y equal-distance entre muchos otros, con NOM. 

\subsection{Simple rules: definition and examples} \label{simple}

When considering the single-peaked domain of preferences, if equal division is taken as a benchmark of fairness, we can single out agents that play a prominent role in the allocation process. In economies with excess demand,  these agents demand less than or equal to equal division and thus free up resources for the rest of the agents. Similarly, in economies with excess supply, these agents demand more than or equal to equal division and thus withhold resources from the rest of the agents. We call these agents ``simple''. We propose the class of ``simple'' rules that reward simple agents by fully satiating them and give each remaining agent an amount between his peak and equal division.\footnote{Several rules in the literature behave in this way. For example, the uniform rule \citep{sprumont1991division} treats non-simple agents as equally as possible, and the sequential rules \citep{barbera1997strategy} allow for an asymmetric treatment of non-simple agents, although imposing a monotonicity condition.}

We can describe a simple rule by means of a two-step procedure. In the first step, simple agents receive their peak (as their final allotment) and the rest receive (provisionally) equal division. In the second step, the amounts assigned to non-simple agents are 
 adjusted to attain feasibility and not overpass the peak amounts: they are increased in economies with excess demand and decreased in economies with excess supply. This second step can be thought of as solving a claims problem\footnote{Remember that a \emph{claims problem} for a set of agents $\overline{N}$ consists of an endowment $E \in \mathbb{R}_+$ and, for each $i \in \overline{N}$, a claim $c_i \in \mathbb{R}_+$ such that $\sum_{j \in \overline{N}}c_j\geq E$.} among non-simple agents. In this problem, the endowment is the (positive) difference between the total amount $\Omega$ and what already has been allocated in the first step and, for each non-simple agent, the claim is the (positive) difference between his peak and equal division. The endowment thus defined guarantees feasibility in the original problem, whereas the claims thus defined guarantee that adjusted amounts do not overpass the peak amounts.\footnote{For a thorough account on these facts, see Lemma \ref{lema simple}.}

Given $(R,\Omega) \in \mathcal{E}_{\mathcal{SP}},$  let $z(R, \Omega)=\sum_{j \in N}p(R_j)-\Omega.$ If $z(R,\Omega) \geq 0$ we say that economy $(R, \Omega)$ has \textbf{excess demand} whereas if $z(R,\Omega)<0$ we say that economy $(R, \Omega)$ has \textbf{excess supply}.\footnote{We deliberately include the \textbf{balanced} case where $z(R, \Omega)=0$ in the excess demand case. No confusion will arise.} Given economy $(R, \Omega) \in \mathcal{E}_{\mathcal{SP}},$ agent $i \in N$ is  \textbf{simple} if either $z(R, \Omega) \geq 0$ and $p(R_i)<\frac{\Omega}{n}$ or $z(R, \Omega)\leq 0$ and $p(R_i)>\frac{\Omega}{n}$. Let $N^+(R, \Omega)$ denote the set of simple agents of economy $(R, \Omega)$ and let $N^-(R,\Omega)=N \setminus N^+(R, \Omega)$ be the set of non-simple agents.  Let $$E(R,\Omega)= \left|\Omega- \left(\sum_{j\in N^+(R, \Omega)} p(R_j)+ |N^-(R, \Omega)| \frac{\Omega}{n}\right)\right|$$ be the amount to be adjusted in the second step of the procedure just described. We are now in a position to define our rules.

 \begin{definition}\label{def simple}
An own-peak-only rule $\varphi$ defined on 
 $\mathcal{E}_\mathcal{SP}$ is  \textbf{simple} if for each $ (R, \Omega) \in \mathcal{E}_{\mathcal{SP}}$ and each $i\in N$,
$$\varphi_i(R, \Omega)=\left\{\begin{array}{l l }
p(R_i) & \text{if } i\in N^+(R, \Omega) \\
\frac{\Omega}{n}+\nu_i & \text{if }  i\in N^-(R, \Omega) \text{ and } z(R, \Omega)\geq 0\\
\frac{\Omega}{n}-\nu_i & \text{if }  i\in N^-(R, \Omega) \text{ and } z(R, \Omega)\leq 0\\
\end{array}\right.$$
where $\nu_i=\nu_i(R,\Omega)$ is such that $0\leq \nu_i\leq |p(R_i)-\frac{\Omega}{n}|$ and  $\sum_{j \in N^-(R, \Omega)}\nu_j=E(R,\Omega).$ 
\end{definition}
Before proceeding, two remarks about defining a simple rule are in order.
First, finding $\nu_i(R,\Omega)$ for each $i\in N^-(R, \Omega)$ is equivalent to solving a claims problem involving all non-simple agents, where $E(R,\Omega)$ is the endowment and $c_i(R,\Omega)=|p(R_i)-\frac{\Omega}{n}|$ is the claim for each non-simple agent $i\in N^-(R, \Omega)$. Note that if $z(R, \Omega)\geq 0$, $c_i(R_i,\Omega)=p(R_i)-\frac{\Omega}{n}$ for each  $i\in N^-(R, \Omega)$. Therefore,\footnote{In what follows, we omit the dependence on $(R, \Omega)$ to ease notation.}
$$\sum_{j\in N^-}c_j+\sum_{j\in N^+} \left(p(R_j)-\frac{\Omega}{n}\right) =\sum_{j \in N} \left(p(R_j)-\frac{\Omega}{n}\right)=z(R, \Omega)\geq 0,$$ implying that $$\sum_{j \in N^-}c_j \geq \sum_{j\in N^+} \left(\frac{\Omega}{n}-p(R_j)\right)= \Omega- \left(\sum_{j \in N^+} p(R_j)+|N^-| \frac{\Omega}{n}\right)= E(R,\Omega).$$ 
Thus, $\sum_{j \in N^-}c_j \geq  E(R,\Omega)$ and the claims problem is well-defined. Therefore the existence of $\nu_i(R,\Omega)$ is guaranteed. A similar check can be done for economies in which  $z(R, \Omega)\leq 0$. 
%,  $c_i(R_i,\Omega)=\frac{\Omega}{n}-p(R_i)$ for each  $i\in N^r$. Therefore, 
%$$\sum_{i\in N^r}c_i(R_i,\Omega)+\sum_{i\in N^{\star}} (\frac{\Omega}{n}-p(R_i)) =\sum_{i\in N} (\frac{\Omega}{n}-p(R_i))=-z(R, \Omega)\geq 0.$$ 
%Then, $$\sum_{i\in N^r}c_i(R_i,\Omega)\geq \sum_{i\in N^{\star}}(p(R_i)-\frac{\Omega}{n})= (\sum_{i\in N^\star(R, \Omega)} p(R_i)+ |N^r| \frac{\Omega}{n})-\Omega= E(R,\Omega).$$
%Therefore, in both cases, 
%$$\sum_{i\in N^r}c_i(R_i,\Omega)\geq E(R,\Omega)$$ 
%and the claim problem is well-defined.
Second, notice that the requirement of own-peak-onliness for $\varphi$ is equivalent to $\nu_i(R_i,R_{-i},\Omega)=\nu_i(R'_i,R_{-i},\Omega)$ whenever $p(R_i)=p(R'_i)$.

%The characterization of simple rules that stated previous lemma will be useful in the proof of our results.  

Notice that once $\nu$ is specified, the associated simple rule $\varphi$ is completely determined. Next, we present some relevant examples.

\begin{example}\label{example claims rules}
Consider following three known claims rules
\begin{itemize}
    
     \item \textbf {constrained equal awards:} for each $(R,\Omega) \in \mathcal{E}_{\mathcal{SP}}$ and each $i\in N^-(R, \Omega)$, 
$$\nu^{CEA}_i(R, \Omega)= \min\left\{\left|p(R_i)-\frac{\Omega}{n}\right|, \lambda\right\} $$
where $\lambda \geq 0$ and solves $\sum_{j \in N^-(R, \Omega)} \nu^{CEA}_j(R, \Omega)=E(R,\Omega).$

    \item \textbf{constrained equal losses:} for each $(R,\omega) \in \mathcal{E}_{\mathcal{SP}}$ and each $i \in N^-(R,\Omega)$,   

$$\nu^{CEL}_i(R, \Omega)= \min\left\{0,\left|p(R_i)-\frac{\Omega}{n}\right|- \lambda\right\} $$
where $\lambda \geq 0$ and solves $\sum_{j \in N^-(R, \Omega)} \nu^{CEL}_j(R, \Omega)=E(R,\Omega),$ and

\item \textbf{proportional:} for each $(R,\omega) \in \mathcal{E}_{\mathcal{SP}}$ and each $i \in N^-(R,\Omega)$,   

$$\nu^{Pro}_i(R,\Omega)= \frac{|p(R_i)-\frac{\Omega}{n}|}{\sum_{j\in N^-(R, \Omega)} |p(R_j)-\frac{\Omega}{n}|} E(R,\Omega).$$
\end{itemize}

It is easy to see that its associated simple rule is no other than the previously defined uniform rule. We call the simple rules associated with $\nu^{CEL}$ and $\nu^{Pro}$ \textbf{constrained equal-distance simple} and \textbf{proportional simple}  rules, respectively. Observe that, by Definition \ref{def simple} and the definition of $\nu^{CEL}$, a non-simple agent receives $\frac{\Omega}{n}$ or $p(R_i) \pm \lambda$ in the constrained equal-distance simple rule. Therefore, this rule effectively seeks to equate the distance between assignments and peak amounts as much as possible among non-simple agents. Similarly, by Definition \ref{def simple} and the definition of $\nu^{Pro}$, for a non-simple agent, the distance between equal division and his assignment in the proportional simple rule is proportional to the distance between equal division and his peak. \footnote{Observe that allocating $E(R,\Omega)$ proportionally with respect to the peaks could assign to some non-simple agent $i$ more than $|p(R_i)-\frac{\Omega}{n}|$. As a consequence, such allocation can not be used to construct a simple rule, because some non-simple agent could obtain an amount that is not between the equal division and his peak.}
\hfill $\Diamond$ 

\end{example}

An algorithmic procedure is provided in Appendix \ref{appendix} as an alternative way to construct a generic simple rule without using claims rules. 

\subsection{Characterizations}\label{subsection characterizations}

%ESTO Y EL LEMA DEBERIA IR A LA SECCION SIGUIENTE)¿¿
Before presenting the characterizations, we introduce a useful result. The following lemma provides an easy description of simple rules. %, that will be useful in proving the results of the next subsection. 
Given $x,y,z \in \mathbb{R}$, we say that \textbf{$\boldsymbol{x$ is between $y$ and $z}$} if $y\leq z$ and $x \in [y,z]$ or $z\leq y$ and $x \in [z,y].$

\begin{lemma}    
\label{lema simple}
An own-peak-only rule $\varphi$ defined on 
 $\mathcal{E}_\mathcal{SP}$ is  simple if and only if for each $ (R, \Omega) \in \mathcal{E}_{\mathcal{SP}}$,
\begin{enumerate}[(i)]
    \item $\varphi_i(R, \Omega)=p(R_i)$ if  $i \in N^+(R, \Omega),$ and
    \item $\varphi_i(R, \Omega)$ is between $\frac{\Omega}{n}$ and $p(R_i)$ for each $i \in N^-(R, \Omega)$.
    % \item $\varphi_i(R, \Omega)\leq \varphi_j(R, \Omega)$ if $p(R_i)\leq p(R_j)$.
\end{enumerate}
\end{lemma}
\begin{proof}
    ($\Longrightarrow$) Let $\varphi$ be a simple rule on $\mathcal{E}_\mathcal{SP}$. Then, (i) is trivially satisfied. To see (ii),  consider an economy $(R, \Omega)\in \mathcal{E}_\mathcal{SP}$ and $i\in N^-(R, \Omega)$. Assume that $z(R, \Omega)\geq 0$ (the other case is symmetric). Then, $|p(R_i)-\frac{\Omega}{n}|=p(R_i)-\frac{\Omega}{n}$. Therefore, as $0\leq \nu_i(R, \Omega)\leq p(R_i)-\frac{\Omega}{n}$ and $\varphi_i(R, \Omega)=\frac{\Omega}{n}+\nu_i(R, \Omega)$, 
    $$\frac{\Omega}{n}\leq \varphi_i(R, \Omega)\leq p(R_i).$$

    \noindent ($\Longleftarrow$) Let $\varphi$ be an own-peak-only rule defined on $\mathcal{E}_\mathcal{SP}$ and satisfying (i) and (ii). %It is sufficient to show the existence of $\varphi$. 
    Let  $(R, \Omega)\in \mathcal{E}_\mathcal{SP}$ and let  $i\in N^-(R, \Omega)$. Assume that $z(R, \Omega)\geq 0$ (the other case is symmetric). Then, $p(R_i) \geq \frac{\Omega}{n}$ and thus,  
     by (ii), 
    \begin{equation} \label{desigual}
        \frac{\Omega}{n}\leq\varphi_i(R, \Omega)\leq p(R_i). 
    \end{equation}
    Define, for each $i \in N^-(R, \Omega)$, $\nu_i(R,\Omega):=\varphi_i(R, \Omega)-\frac{\Omega}{n}$. By (\ref{desigual}), $$0\leq\nu_i(R,\Omega)\leq p(R_i)-\frac{\Omega}{n}=\left|p(R_i)-\frac{\Omega}{n}\right|.$$ 
    Furthermore, by feasibility and (i),   
    $\sum_{j \in N^+(R, \Omega)}  p(R_j)+\sum_{j\in N^-(R,\Omega)}\varphi_j(R, \Omega)= \Omega$. Therefore,  $\sum_{j \in N^-(R, \Omega)}(\nu_j(R,\Omega)+\frac{\Omega}{n})= \Omega-\sum_{j\in N^+(R, \Omega)} p(R_j)$, which in turn is equivalent to 
    $$\sum_{j \in N^-(R,\omega)}\nu_j(R,\Omega))= \Omega-\left(\sum_{j\in N^+(R,\Omega)} p(R_j)+|N^-(R,\Omega)|\frac{\Omega}{n}\right)=E(R,\Omega)$$
Therefore, $\varphi$ is a simple rule, and $\nu$ is the corresponding claims rule associated with $\varphi$.
\end{proof}

The following result shows that simple rules are not obviously manipulable. 

\begin{theorem} \label{t1} A simple rule defined on $\mathcal{E}_{\mathcal{SP}}$ satisfies NOM. 
\end{theorem}
\begin{proof} Let $\varphi$ be a simple rule defined on $\mathcal{E}_\mathcal{SP}$, and let $i\in N$ and $(R_i, \Omega) \in \mathcal{SP} \times \mathbb{R}_{++}$.
%The next claim about the option sets is necessary

\noindent \textbf{Claim: $\boldsymbol{O^\varphi(R_i,\Omega)=\begin{cases}
    \left[\frac{\Omega}{n}, p(R_i)\right] & \text{ if } p(R_i)>\frac{\Omega}{n} \\
    \left[p(R_i), \frac{\Omega}{n}\right] & \text{ if } p(R_i)\leq \frac{\Omega}{n}\\
\end{cases}}$} 

\noindent Assume that $\frac{\Omega}{n} <p(R_i) $ (the other case is symmetric).  That $O^\varphi(R_i, \Omega) \subseteq [\frac{\Omega}{n}, p(R_i)]$ is clear by the definition of the rule. Now, let $x\in [\frac{\Omega}{n},p(R_i)]$. We will prove that $x \in O^\varphi(R_i, \Omega)$. Consider $R_{-i} \in \mathcal{SP}^{n-1}$ such that $p(R_{j})=\frac{\Omega}{n}-\frac{x-\frac{\Omega}{n}}{n-1}$ for each  $j\in N\setminus\{i\}$. Then, as $\frac{\Omega}{n}\leq x$, $p(R_{j})\leq\frac{\Omega}{n}\leq x$ for all $j\in N\setminus\{i\}$ and $\sum_{j\in N\setminus\{i\} }p(R_j)+x=\Omega$. By definition of the rule, $\varphi_j(R, \Omega)=p(R_{j})$ for all $j\in N\setminus\{i\}$. Then, by feasibility, $\varphi_i(R, \Omega)=x$. Therefore, $[\frac{\Omega}{n} , p(R_i)]\subseteq O^\varphi(R_i,\Omega)$. This completes the proof of the Claim.

To see that $\varphi$ satisfies NOM, assume that  $R_i'$ is a manipulation of $\varphi$ at $(R_i, \Omega)$. By the Claim, $\frac{\Omega}{n} \in O^\varphi(R_i', \Omega)$. Furthermore, by  the Claim and single-peakedness, $xR_i\frac{\Omega}{n}$ for each $x\in O^\varphi(R_i)$. Then, $R_i'$ is not an obvious manipulation and so $\varphi$ is NOM. 
\end{proof}

%\begin{theorem} \label{characterization}
 %An own-peak-only rule defined on  $\mathcal{E}_\mathcal{SP}$ satisfies EFF, EDG, and NOM  if and only if it is a simple rule. 
%\end{theorem}

In Proposition \ref{negative} we saw that both the constrained equal-distance and the proportional criteria are incompatible with NOM when all agents are considered. However, as can be seen in Example \ref{example claims rules}, we can construct simple rules that (at least partially) recover those principles and at the same time, by Theorem \ref{t1}, are compatible with NOM.\footnote{Notice that the incompatibility with NOM remains if a rule assigns the peak amount to simple agents and applies either the constrained equal-distance or the proportional criteria to non-simple agents \emph{without} provisionally assigning equal division to them in a first step. Actually, this is what we do in the examples within the proof of Proposition \ref{negative}.}

Next, we present our first characterization result. 

\begin{theorem} \label{characterization}
 A rule defined on  $\mathcal{E}_\mathcal{SP}$ satisfies own-peak-onliness, efficiency, the equal division guarantee, and  NOM  if and only if it is a simple rule. 
\end{theorem}
\begin{proof}   
\noindent $(\Longleftarrow)$ Let $\varphi$ be a simple rule on $\mathcal{E}_\mathcal{SP}$.  Then, $\varphi$ is trivially own-peak-only and satisfies the equal division guarantee. By definition, $\varphi$ satisfies same-sidedness. Therefore, by Remark \ref{same eff}, $\varphi$ is efficient. Lastly, by Theorem \ref{t1},  $\varphi$ satisfies NOM.

\noindent $(\Longrightarrow)$ Let $\varphi$ be an own-peak-only rule defined on $\mathcal{E}_\mathcal{SP}$  that satisfies efficiency, the equal division guarantee, and NOM. Using Lemma \ref{lema simple}, we now prove that $\varphi$ is simple. Let $(R, \Omega) \in \mathcal{E}_{\mathcal{SP}}$ and assume that $z(R, \Omega) \geq 0$. Let $i\in N$. By efficiency, $\varphi_i(R, \Omega)\leq p(R_i)$. There are two cases to consider: 

\begin{enumerate}
    \item[$\boldsymbol{1}.$] \textbf{$\boldsymbol{p(R_i)\leq \frac{\Omega}{n}}$}. We need to show that $\varphi_i(R, \Omega)\geq p(R_i)$ also holds. Assume, by way of contradiction,  that $\varphi_i(R, \Omega)<p(R_i).$ By own-peak-onliness, we can assume that  $R_i$ is such that 
    %\begin{equation}\label{eq4}
    $\frac{\Omega}{n} \ P_i \ \varphi_i(R, \Omega).$ 
    %\end{equation} 
    Let $R_i' \in \mathcal{SP}$ be such that $p(R_i')=\frac{\Omega}{n}$. By the equal division guarantee, $\varphi_i(R_i',R_{-i}, \Omega)=\frac{\Omega}{n}$. Hence,
    \begin{equation}\label{eq4bis}
    \varphi_i(R_i',R_{-i}, \Omega)=\frac{\Omega}{n} \ P_i \ \varphi_i(R, \Omega)
    \end{equation}
    and $R_i'$ is a manipulation of $\varphi$ at $(R_i, \Omega)$. Furthermore, by the equal division guarantee, $O^\varphi(R_i', \Omega)=\left\{\frac{\Omega}{n}\right\}$. As $\varphi_i(R, \Omega) \in O^\varphi(R_i, \Omega),$ by \eqref{eq4bis},  $R_i'$ is an obvious manipulation of $\varphi$ at $(R_i, \Omega)$, contradicting that $\varphi$ is NOM. Thus, $\varphi_i(R, \Omega)\geq p(R_i)$ and, since by hypothesis $\varphi_i(R, \Omega)\leq p(R_i)$, we have $\varphi_i(R, \Omega)= p(R_i)$.
    
    \item[$\boldsymbol{2}.$] \textbf{$\boldsymbol{p(R_i)> \frac{\Omega}{n}}$}. We need to show that $\varphi_i(R, \Omega)\geq \frac{\Omega}{n}$. Assume, by way of contradiction,  that $\varphi_i(R, \Omega)<\frac{\Omega}{n}.$ By single-peakedness and own-peak-onliness, $\frac{\Omega}{n} \ P_i \ \varphi_i(R, \Omega)$ also holds here, and the proof follows an argument similar to that of the previous case. 
    
\end{enumerate}
We have proved that if $p(R_i)\leq \frac{\Omega}{n}$, $\varphi_i(R, \Omega)=p(R_i)$; whereas if $p(R_i)> \frac{\Omega}{n},$  $\varphi_i(R, \Omega)\in [\frac{\Omega}{n},p(R_i)]$. To complete the proof that $\varphi$ is a simple rule, a symmetric argument can be used for the case $z(R, \Omega)<0$. 
\end{proof}

%\textcolor{red}{No deberíamos incluir una regla simple natural que se use o pueda llegar a tener aplicación en la práctica que sea NOM \emph{pero} manipulable?}

Simple rules only meet a minimal fairness requirement: whenever an agent has equal division as his peak consumption, the rule must assign it to him. Despite this requirement, simple rules can %be quite unfair and 
fail symmetry. If we want to add symmetry to the picture, then the next result follows.\footnote{As we will see in Subsection \ref{subsection independence}, the equal division guarantee and symmetry are independent axioms.}

%\textcolor{blue}{Bajo eff, simetria y SP sólo la regla uniformed (PONER CITA). La pregunta natural que tratamos de responder con el siguiente corolario es qué otras reglas acompañan a la uniforme cuando relajamos SP y asumimos agentes limitados que sólo realizan manipulaciones obvias.  El siguiente corolario queue siguen inmediatamente del Theorem \ref{characterization} nos permite mostrar una la amplia clase de reglas symetric simple rules acompañan a la uniforme cuando asumimos agentes limitados que sólo realizan manipulaciones obvias}

\begin{corollary} \label{characterization 2}
 A rule defined on  $\mathcal{E}_\mathcal{SP}$ satisfies own-peak-onliness, efficiency, the equal division guarantee, symmetry, and NOM if and only if it is a symmetric and simple rule. 
\end{corollary}
\cite{ching1994alternative} shows that if we require strategy-proofness instead of NOM in Corollary \ref{characterization 2}, the unique rule that prevails is the uniform rule. Therefore, Corollary \ref{characterization 2} evidences how the family of characterizing rules enriches when NOM is sufficient to describe the strategic behavior of the agents.
Two interesting members in this family are the simple rules associated with the claims rules of Example \ref{example claims rules}. These rules are compelling examples of rules that can be relevant in some situations but are completely overlooked when the strong incentive compatibility requirement of strategy-proofness is imposed.

\begin{remark} In order to construct a symmetric and simple rule, it is sufficient to consider a symmetric claims rule, i.e., a claims rule that for each problem assigns the same amounts to agents with the same claims. It is straightforward that a simple rule associated with a symmetric claims rule is symmetric. Notice also that symmetric claims rules abound. For example,  $\nu^{CEA},$ $\nu^{CEL}$, and $\nu^{Pro}$ are all symmetric.  
%\textcolor{red}{COMO CONSTRUIR UNA REGLA SIMPLE SIMETRICA}
 %   More generally, notice that for any symmetric claims rule,\footnote{A claims rule is \textbf{symmetric} if, in each problem, it assigns the same amounts to agents with the same claims. Symmetric claims rules abound. For example,  $\nu^{CEA},$ $\nu^{CEL}$, and $\nu^{Pro}$ are all symmetric.}  its associated simple rule is also symmetric. 
\end{remark} 

%EN GENERAL PARA CUALQUIR LAIM RULE SIMETRICA, \footnote{PONER EJEMPLOS DE CLAIMS SIMETRICAS} TENDREMOS QUE SU ASOCIADA SIMPLE RULE SERA SIMETRICA Y POR LO TANTO EL ENRIQUESIMINETO DE LA CLASES SE HACE MUY NOTORIO.

\subsection{Independence of axioms}\label{subsection independence}

Next, we present some examples that prove the independence of the axioms in Theorem \ref{characterization} and Corollary \ref{characterization 2}.

\begin{itemize}
    \item Let $\widetilde{\varphi}:\mathcal{E}_\mathcal{SP} \longrightarrow \mathbb{R}_+^n$ be such that, for each $(R,\Omega)\in \mathcal{E}_\mathcal{SP}$ and each $i\in N$, $\varphi_i(R,\Omega)=\frac{\Omega}{n}$ . Then $\widetilde{\varphi}$ satisfies all properties but efficiency. 

\item Let $\varphi^\star:\mathcal{E}_\mathcal{SP} \longrightarrow \mathbb{R}_+^n$ be such that, for each $(R,\Omega)\in \mathcal{E}_\mathcal{SP}$ and each $i\in N$, 
$$\varphi^\star(R, \Omega)=\left\{\begin{array}{l l }
(p(R_1), p(R_2), 0, \ldots,0)) & \text{if } p(R_1)+p(R_2)=\Omega \text{ and } \\ & p(R_j) \notin \{p(R_1),p(R_2)\} \text{ for each }j \in N \setminus\{1,2\} 
\\
u(R, \Omega) & \text{otherwise}\\
\end{array}\right.$$

Observe that this rule does meet 
the equal division guarantee. It is clear that $\varphi^\star$ is own-peak-only, efficient, and symmetric. Now, we prove that $\varphi^\star$ also satisfies NOM. Let  $i\in N$ and  $(R_i,\Omega)\in \mathcal{SP}\times \mathbb{R}_+$. Notice that  $O^{\varphi^\star}(R_i,\Omega)=O^u(R_i,\Omega)\cup \{0\}$ if  $i \in N \setminus\{1,2\}$ and $O^{\varphi^\star}(R_i,\Omega)=O^u(R_i,\Omega)$ if $i \in \{1,2\}$. Then, the fact that $\varphi^\star$ does not have obvious manipulations for agents in $N\setminus\{1,2\}$ follows a similar argument as in Theorem \ref{t1} using the fact that $u$ is a simple rule. Then, $\varphi^\star$ satisfies all properties but the equal division guarantee.

\item Let $\overline{\varphi}:\mathcal{E}_\mathcal{SP} \longrightarrow \mathbb{R}_+^n$ be such that, for each $(R,\Omega)\in \mathcal{E}_\mathcal{SP}$,   
$$\overline{\varphi}(R, \Omega)=\left\{\begin{array}{l l }
(\frac{1}{3}\Omega, \frac{2}{3}\Omega, 0, \ldots,0) & \text{if } p(R_1)=p(R_2)=\Omega \text{ and } \\ & p(R_j) = 0 \text{ for each }j \in N \setminus\{1,2\} \\
u(R, \Omega) & \text{otherwise}\\
\end{array}\right.$$

Observe that $\overline{\varphi}$  is simple. Thus,  it satisfies own-peak-onliness, efficiency, the equal division guarantee, and NOM. However, it does not satisfy symmetry.

\item Let $\widehat{\varphi}:\mathcal{E}_\mathcal{SP} \longrightarrow \mathbb{R}_+^n$ be such that, for each $(R,\Omega)\in \mathcal{E}_\mathcal{SP}$,  
\begin{enumerate}[(i)]
    \item if $z(R,\Omega)\geq 0, \widehat{\varphi}(R, \Omega)=u(R,\Omega)$.
     \item if $z(R,\Omega)< 0$, let $\widehat{N}=\{i\in N: p(R_i)=\min_{j\in N}\{p(R_j)\}\}$. Then, for each $i\in N$,

$$\widehat{\varphi}_i(R, \Omega)=\left\{\begin{array}{l l }
p(R_i) & \text{if } i \in N\setminus \widehat{N} %\text{ and } \lambda \\ & \text{and }p(R_1)<p(R_j) \text{ for each }j \in N \setminus\{1\} 
\\
\lambda & \text{if } i \in \widehat{N}\\
\end{array}\right.$$
where $\lambda \geq 0$ and solves $\sum_{j \in N}\widehat{\varphi}_j(R, \Omega)=\Omega$.
%For each $(R,\Omega)\in\mathcal{E}_\mathcal{SP},$ and each $i \in N,$
\end{enumerate}

Observe that $\widehat{\varphi}$  is not simple because it could be that $p(R_i)<\frac{\Omega}{n}<\widehat{\varphi}_i(R)$ for some $i\in \widehat{N}$. It is easy to see that $\widehat{\varphi}$ is own-peak-only, efficient, and symmetric. If $z(R,\Omega)< 0$ and $p(R_i)=\frac{\Omega}{n}$, then $i \in N\setminus \widehat{N}$. Therefore, $\widehat{\varphi}$ satisfies the equal division guarantee. This implies, by Theorem \ref{characterization},  that $\widehat{\varphi}$ satisfies all properties but NOM. 

\item  Let $\underline{\varphi}:\mathcal{E}_\mathcal{SP} \longrightarrow \mathbb{R}_+^n$ be such that, for each $(R,\Omega)\in \mathcal{E}_\mathcal{SP}$ and each $i\in N$,  
$$\underline{\varphi}_i(R, \Omega)=\left\{\begin{array}{l l }
u_i(R^0_1,R_{-1},\Omega) & \text{if } z(R_{-1}, \Omega)\geq 0,  \ 0P_1\frac{\Omega}{n}, \text{ and } \\ & p(R_1)<p(R_j) \text{ for each }j \in N \setminus\{1\} \\
u_i(R, \Omega) & \text{otherwise}\\
\end{array}\right.$$
where $R^0_1\in \mathcal{SP}$ is such that $p(R^0_1)=0$.

Observe that $\underline{\varphi}$ is not  simple because it is not own-peak-only and it could be that $\varphi_1(R, \Omega)<p(R_1)<\frac{\Omega}{n}$. It is clear that $\underline{\varphi}$ satisfies the equal division guarantee and symmetry. To see that $\underline{\varphi}$ is efficient, let $(R,\Omega)\in \mathcal{E}_\mathcal{SP}$ be such that  $z(R_{-1}, \Omega)\geq 0$, $0P_1\frac{\Omega}{n}$, and $p(R_1)<p(R_j) \text{ for each }j \in N \setminus\{1\}$. Then, by the definition of $u$ and the fact that $u$ satisfies same-sideness, $\underline{\varphi}_1(R, \Omega)=0\leq p(R_1)$ and $\varphi_i(R, \Omega)\leq p(R_i)$ for each $i\in N\setminus\{1\}$. Then, $\underline{\varphi}_i(R, \Omega)\leq p(R_i)$  for each $i\in N$. Since in any other case $\underline{\varphi}$ operates as $u$ and $u$  satisfies same-sideness, $\underline{\varphi}$ also satisfies this property. Thus, by Remark \ref{same eff}, $\underline{\varphi}$ is efficient. Now, we prove that $\underline{\varphi}$ satisfies NOM. For each $i\in N\setminus\{1\}$ and each $(R_i,\Omega)\in \mathcal{SP}\times \mathbb{R}_+$,  $O^{\underline{\varphi}}(R_i,\Omega)=O^u(R_i,\Omega)$. Then, the fact that $\underline{\varphi}$ does not have obvious manipulations for agents in $N\setminus\{1\}$ follows the same arguments used in the proof of  Theorem \ref{t1} and the fact that $u$ is a simple rule. Let us next consider agent $1$. First, let $(R_1,\Omega)\in \mathcal{SP}\times \mathbb{R}_+$ be  such that $\frac{\Omega}{n}R_10$. Then, $O^{\underline{\varphi}}(R_1,\Omega)=O^u(R_1,\Omega)$.  Again, the fact that $\underline{\varphi}$ does not have obvious manipulations for agent $1$ at $(R_1,\Omega)$ follows the same arguments used in the proof of  Theorem \ref{t1} together with the facts that $u$ is a simple rule and that $\frac{\Omega}{n}\in O^\varphi(R'_1,\Omega)$ for each $R'_1 \in \mathcal{SP}$. Second, let $(R_1,\Omega)\in \mathcal{SP}\times \mathbb{R}_+$ be such that $0P_1\frac{\omega}{n}$. Then, by single-peakedness of $R_1$, $p(R_1)<\frac{\Omega}{n}.$ By definition of $\underline{\varphi}$,  $O^{\underline{\varphi}}(R_1,\Omega)=O^u(R^0_1,\Omega)=[0,\frac{\Omega}{n}]$. Using  single-peakedness again and the fact that $0P_1\frac{\Omega}{n}$, it follows that  $xR_1\frac{\Omega}{n}$ for each $x\in O^\varphi(R_1,\Omega)$. Furthermore, as $\frac{\Omega}{n}\in O^\varphi(R'_1,\Omega)$ for each $R'_1 \in \mathcal{SP}$, agent $1$ does not have  obvious manipulations of $\underline{\varphi}$ at $(R_1,\Omega).$
Therefore, $\underline{\varphi}$ is NOM. We conclude that $\underline{\varphi}$ satisfies all properties but own-peak-onliness. 
    
\end{itemize}
    
%\footnote{Since this rule  operates on the single-peaked domain, it is own-peak-only.}

\section{Further characterizations involving NOM}\label{section further}

Insisting on %efficiency and 
the NOM property, in this section we provide new characterizations of rules 
in which %efficiency and 
NOM together with other relevant properties are invoked. These considerations lead to two interesting variants of Theorem \ref{characterization}. %Insisting on efficiency as well, %The first one involves the peaks-only property, 
The first one is an adaptation of Theorem \ref{characterization} %of the equal division guarantee  
to economies with individual endowments and the second one involves a responsiveness condition concerning the peaks of the preference profile. %Lastly, the third one gets rid of efficiency. 

%\subsection{Peaks-onliness}

\subsection{Endowments guarantee}

In some situations,  it is more appropriate to assume that instead of a social endowment $\Omega \in \mathbb{R}_+$, there is a profile  $\omega=(\omega_i)_{i \in N} \in \mathbb{R}_+^n$  where, for each $i \in N$, $\omega_i$ denotes agent $i$'s individual endowment of the non-disposable commodity and  $\sum_{j \in N}\omega_j=\Omega$.  
%The third remark is with respect to EDG and equal division. In some papers, it is assumed that each agent $i\in N$ has an initial endowment $q_i \in \mathbb{R}_+ $ and that the social endowment is $\Omega=\sum_{i\in N} q_i$.  
In such cases, it is natural to replace the equal division guarantee with the following property.  %that guarantees $q_i$, for each $i\in N$.

\vspace{5 pt}
\noindent
\textbf{Endowments guarantee:} For each   $(R,\omega) \in \mathcal{E}_{\mathcal{D}}$ and each $i \in N$ such that $\omega_i\in p(R_i)$, we have $\varphi_i(R)I_i\omega_i$.
\vspace{5 pt}

\noindent If, for each $i \in N$, we replace $\frac{\Omega}{n}$ with $\omega_i$ in the definition of simple agent, we can define, following Lemma \ref{lema simple}, a \textbf{simple reallocation rule} as one that assigns: (i) to each simple agent, his peak amount and, (ii) to each remaining agent, an amount between his endowment and his peak.  Then, for economies with individual endowments, the following variant of Theorem \ref{characterization} is obtained.

\begin{proposition} \label{characterization bis}
 A rule defined on  $\mathcal{E}_\mathcal{SP}$ satisfies own-peak-onliness, efficiency, the endowments guarantee, and NOM  if and only if it is a simple reallocation rule. 
\end{proposition}

The analysis of the independence of axioms involved is similar to that of Theorem \ref{characterization} (see Subsection \ref{subsection independence}).

%Theorem \ref{characterization} and Corollary  \ref{characterization 2} hold replacing the equal division guarantee with the endowments guarantee. 

\subsection{Peak responsiveness}

The following responsiveness property says that if the peak of one agent is greater than or equal to the peak of another agent, then the first agent must be assigned an amount greater than or equal to that of the second agent. 

%The fourth remark is with respect to EDG and SYM. Although EDG and SYM are relational to fairness criteria both are independent as we see in Section \ref{simple}. Next, we present a monotonicity axiom that, under NOM, encompasses EDG and SYM and allows us to present a variant version of Theorem \ref{characterization} and Corollary \ref{characterization 2}

%Given $R_i,R_j \in \mathcal{SP}$, we say that $\boldsymbol{p(R_i) \leq p(R_j)}$ whenever for each $x \in p(R_i)$ and each $y \in p(R_j)$ we have $x \leq y$.
\vspace{5 pt}
\noindent
\textbf{Peak responsiveness:} For each  $(R, \Omega) \in \mathcal{E}_\mathcal{SP}$ and each $\{i,j\} \subseteq N$ with $i\neq j$, $p(R_i) \leq p(R_j)$ implies  $\varphi_i(R, \Omega)\leq \varphi_j(R, \Omega).$
\vspace{5 pt}

%\textcolor{red}{peak responsiveness???}

\noindent It is clear that peak responsiveness implies symmetry.

Furthermore, we next prove that peak responsiveness together with NOM imply the equal division guarantee for an own-peak-only rule defined on $\mathcal{E}_\mathcal{SP}$.

\begin{lemma} \label{PMEDG}
Any rule defined on $\mathcal{E}_\mathcal{SP}$ that satisfies own-peak-onliness, peak responsiveness, and NOM meets the equal division guarantee.
    \end{lemma}
%\noindent \textbf{Claim: If $\boldsymbol{ i\in N$ and $R\in \mathcal{R}^n$ are such that $p(R_i)=\frac{\Omega}{n}$, then $\varphi_i(R)=\frac{\Omega}{n}}$  } 
\begin{proof}
    Let $\varphi$ be an own-peak-only rule defined on $\mathcal{E}_\mathcal{SP}$ that satisfies peak responsiveness and NOM. Assume, on contradiction, that there are $(R, \Omega) \in \mathcal{E}_{\mathcal{SP}}$ and  $i\in N$  such that $p(R_i)=\frac{\Omega}{n}$ and $\varphi_i(R, \Omega)\neq\frac{\Omega}{n}$. Consider the case  $\varphi_i(R, \Omega)<\frac{\Omega}{n}$ (the other one is symmetric). Let $R_i'\in \mathcal{SP}$ be such that $p(R'_i)=\infty$. By peak responsiveness, for each $\widetilde{R}_{-i} \in \mathcal{SP}^{n-1}$ we have that $\varphi_i(R'_i,\widetilde{R}_{-i}, \Omega)\geq \varphi_j(R'_i,\widetilde{R}_{-i}, \Omega)$ for each $j\in N\setminus\{i\}$. Then, by feasibility, $\varphi_i(R'_i,R_{-i}, \Omega)\geq \frac{\Omega}{n}$. By own-peak-onliness, we can assume that $R_i$ is such that   $\varphi_i(R'_i,R_{-i}, \Omega)P_i\varphi_i(R, \Omega)$. Then, $R'_i$ is a manipulation of $\varphi$ at $(R_i,\Omega)$. Furthermore, if $x\in O^\varphi(R'_i,\Omega)$, then  $x=\varphi_i(R'_i,R'_{-i}, \Omega)$ for some $R'_{-i} \in \mathcal{SP}^{n-1}$. Then, by peak responsiveness and feasibility, $x\geq \frac{\Omega}{n}$. Therefore, by own-peak-onliness, we can assume that $xR_i\varphi_i(R, \Omega)$. Therefore, $R'_i$ is an obvious manipulation of $\varphi$ at $(R,\Omega)$, which is a contradiction.
\end{proof} 

Hence, under own-peak-onliness and NOM, peak responsiveness implies the fairness properties of symmetry and the equal division guarantee.\footnote{Remember that, as we saw in Section \ref{simple},  symmetry and the equal division guarantee are independent properties.} By using Lemma \ref{PMEDG} and a proof similar to that of Theorem \ref{characterization}, the following characterization is obtained.

\begin{proposition} \label{characterization3}
 A rule defined on  $\mathcal{E}_\mathcal{SP}$ satisfies own-peak-onliness, efficiency, peak responsiveness, and NOM  if and only if it is a peak responsive and simple rule.\footnote{Notice that the equal division guarantee is never invoked in this proposition.} 
\end{proposition}

\noindent Let us check the independence of axioms involved in Proposition \ref{characterization3}. Considering the rules defined in Subsection \ref{subsection independence}, the equal division rule $\widetilde{\varphi}$  satisfies all properties but efficiency; by Lemma \ref{PMEDG}, rule $\varphi^\star$ satisfies all properties but peak responsiveness; and rule $\underline{\varphi}$ satisfies all properties but own-peak-onliness.   
%The independence of EFF with respect to NOM and PM and the independence of PM with respect to EFF and NOM follow as in Theorem \ref{characterization}, replacing PM with EDG. 
Lastly, the Constrained equal-distance rule is own-peak-only by definition; it is same-sided, and therefore efficient by Remark \ref{same eff}; and it is also peak responsive. However, since it may not assign the peak amount to some simple agent, it is not NOM.% satisfies  EFF and PM but fails NOM and shows the independence of NOM with respect to EFF and PM in the previous Theorem.

\begin{remark} In order to construct a peak responsive and simple rule, it is sufficient to consider a responsive claims rule. A claims rule $\nu$ is \textbf{responsive} if for each problem $(c,E)$ and each $\{i,j\} \subset N$ with $i \neq j$, $c_i \leq c_j$ implies $\nu_i(c,E) \leq \nu_j(c,E)$. To see that the simple rule $\varphi$ associated with a responsive claims rule $\nu$ is peak responsive, let $(R, \Omega) \in \mathcal{E}^N$ and $\{i,j\} \subseteq N$ with $i\neq j$ and $p(R_i) \leq p(R_j)$. %first consider $(R, \Omega) \in \mathcal{E}^N$ such that $z(E,\Omega) \geq 0$ and $\{i,j\} \subseteq N$ with $i\neq j$ and $p(R_i) \leq p(R_j)$. If one or both agents in $\{i,j\}$ are simple, then clearly $\varphi_i(R, \Omega) \leq \varphi_j(R,\Omega)$. If both are non-simple, then $c_i(R, \Omega) \leq c_j(R, \Omega)$ and therefore $\nu_i(R, \Omega) \leq \nu_j(R, \Omega)$. Thus, by Definition \ref{def simple}, $\varphi_i(R, \Omega) \leq \varphi_j(R, \Omega)$. Next, 
Consider first the case  $z(E,\Omega) < 0$. If one or both agents in $\{i,j\}$ are simple, then clearly $\varphi_i(R, \Omega) \leq \varphi_j(R,\Omega)$. If both are non-simple, then $c_i(R, \Omega)=\frac{\Omega}{n} - p(R_i) \geq \frac{\Omega}{n} - p(R_j)= c_j(R, \Omega)$ and therefore $\nu_i(R, \Omega) \geq \nu_j(R, \Omega)$. Thus, by Definition \ref{def simple}, $\varphi_i(R, \Omega) \leq \varphi_j(R, \Omega)$. The case $z(E,\Omega) \geq 0$ is even more straightforward and we omit it. Notice that all claims rules in Example \ref{example claims rules} are responsive, and so their associated simple rules are peak responsive.

%\textcolor{red}{COMO CONSTRUIR UNA REGLA SIMPLE MONOTONA }

%PARA TENER UNA UNA SIMPLE MONOTONA ES LO MISMO QUE TENER UNA CLAIM MONOTONO PONER ARGUMENTO

%SI HAY EXESO DE DEMANDA UN AGENTE CON UN PEAK MAS GRANDE TIENE UN CLAIM MAYOR Y POR TANTO, SI LA CLAIM ES MONOTONA, LA ASIGACION FINAL EN LA REGLA SIMPLE FINAL SERA MAYOR.

%SI HAY EXESO DE OFERTA UN AGENTE CON UN PEAK MAS PEQUEÑO TIENE UN CLAIM MAYOR Y POR TANTO, SI LA CLAIM ES MONOTONA, LA ASIGACION FINAL EN LA REGLA SIMPLE FINAL SERA MENOR (VER DEFICNICION 2)

%LAS REGLAS EN EL EJEMPLO 1 SON EJEMPLOS DE REGLAS MONOTONAS
    
\end{remark}

\section{Maximal domain for NOM}\label{section maximal}

Now, a relevant question that arises is whether the domain of preferences can be enlarged maintaining the compatibility of the properties. \cite{ching1998maximal} show that the single-plateaued domain  is maximal for efficiency, symmetry, and strategy-proofness. As we have seen, when we weaken strategy-proofness to NOM (and explicitly invoke own-peak-onliness), the class of rules that also meet the equal division guarantee enlarges considerably. One may suspect that the maximal domain of preferences for these new properties enlarges as well. However, as we will see in this section, the single-plateaued domain is still maximal for these properties.\footnote{\cite{ching1998maximal} show that the single-plateaued domain is \emph{the} unique maximal domain for their properties. We show that this domain is \emph{one} maximal domain for our properties. There could exist ``pathological'' maximal domains containing just a portion of the single-plateaued domain. Nevertheless, all of them have to be contained in the domain of convex preferences.}    

%ching demuestra que el dominion maximal para s-p, eff y sym es sinlge-plateaued

%cuando debilitamos sp a nom la clase crece enormemente 

%Uno podria sospechar que el dominion maximal tmb crece, pero como veremos, el single plateaued sigue siendo maximal.\footnote{ es el unico dominio. Nosotros demostramos que el single-plateaued es un dominio maximal. Eventualmente podrían existir otros dominions maximales, pero siempre contenidos en el dominio convexo.}

We start with the definition of maximality for a list of properties.

\begin{definition}\label{def max}
A domain $\mathcal{D}^{\star } \subseteq \mathcal{U}$ is \textbf{maximal for a list of properties} if (i) $\mathcal{SP} \subseteq \mathcal{D}^{\star}$, (ii) there is a rule defined on $\mathcal{E}_{\mathcal{D}^{\star}}$ satisfying the properties, and (iii) for each $\mathcal{E}_{\overline{\mathcal{D}}}$ with $\mathcal{D}^\star \subsetneq \overline{\mathcal{D}} \subseteq \mathcal{U}$  there is no rule defined on  $\mathcal{E}_{\overline{\mathcal{D}}}$ satisfying the same properties.
\end{definition}

Given $i \in N$ and $R_i\in \mathcal{U}$,  let $\underline{p}(R_i)=\inf p(R_i)$ and $\overline{p}(R_i)=\sup p(R_i).$
Agent $i$'s preference $R_i \in \mathcal{U}$ is \textbf{single-plateaued} if $p(R_i)$ is the interval $[ \underline{p}(R_i), \overline{p}(R_i)]$ and, for each pair $\{x_i, x_i'\} \subseteq \mathbb{R}_+$, we have $x_iP_ix_i'$ as long as either $x_i'<x_i\leq \underline{p}(R_i)$ or $\overline{p}(R_i) \leq x_i<x_i'$ holds. Denote by $\mathcal{SPL}$ the domain of all such preferences.  Agent $i$'s preference $R_i \in \mathcal{U}$ is \textbf{convex} if $p(R_i)$ is the interval $[\underline{p}(R_i), \overline{p}(R_i)]$ and, for each pair $\{x_i, x_i'\} \subseteq \mathbb{R}_+$, we have $x_iR_ix_i'$ as long as either $x_i'<x_i\leq \underline{p}(R_i)$ or $\overline{p}(R_i) \leq x_i<x_i'$ holds. Denote by $\mathcal{C}$ the domain of all such preferences.

%\textcolor{red}{Explicar que seguimos de cerca las ideas de Ching Serizawa}

\begin{lemma} \label{between}
Let $\mathcal{D}$ be such that $\mathcal{SP}\subseteq \mathcal{D} \subseteq \mathcal{U}$ and let $\varphi$ be an own-peak-only rule defined on $\mathcal{E}_\mathcal{D}$ that satisfies the equal division guarantee and NOM. Let $i \in N$ and $(R, \Omega) \in \mathcal{E}_\mathcal{D}$. If $p(R_i)$ is a singleton, then $\varphi_i(R, \Omega)$ is between $p(R_i)$ and  $\frac{\Omega}{n}$.   
\end{lemma}
\begin{proof}
Let $\varphi:\mathcal{E}_\mathcal{D} \longrightarrow \mathbb{R}_+$ be an own-peak-only rule that satisfies the equal division guarantee and NOM.  Let $(R, \Omega) \in\mathcal{E}_\mathcal{D}$ and  let $i\in N$ be such that $p(R_i)\leq \frac{\Omega}{n} $. We will prove that $\varphi_i(R, \Omega) \in [p(R_i), \frac{\Omega}{n}].$ By own-peak-onliness and the fact that $\mathcal{SP}\subseteq \mathcal{D}$, the proof that  $\varphi_i(R, \Omega)\geq p(R_i)$ follows the same lines as Case 1 in  part  ($\Longrightarrow$) of the proof of Theorem  \ref{characterization}. 
Next, we need to show that $\varphi_i(R, \Omega)\leq \frac{\Omega}{n}$. Assume, by way of contradiction,  that $\frac{\Omega}{n}< \varphi_i(R, \Omega)$. By own-peak-onliness and the fact that $\mathcal{SP}\subseteq \mathcal{D}$, we can assume w.l.o.g. that $\frac{\Omega}{n} \ P_i \ \varphi_i(R, \Omega)$. Let $R'_i \in \mathcal{SP}$ be such that $p(R'_i ) =\frac{\Omega}{n}.$ By the equal division guarantee, $R'_i$ is an obvious manipulation of $\varphi$ at $(R_i,\Omega)$. This contradicts the fact that $\varphi$ satisfies NOM. Therefore, $\varphi_i(R, \Omega) \in [p(R_i), \frac{\Omega}{n}].$ 

The case $\frac{\Omega}{n} > p(R_i)$ is symmetric, so we omit it. \end{proof}

\begin{lemma} \label{convex}
Let $\mathcal{D}$ be such that $\mathcal{SP}\subseteq \mathcal{D} \subseteq \mathcal{U}$. If there is an own peak-only rule defined on $\mathcal{E}_\mathcal{D}$ that satisfies efficiency, the equal division guarantee, and NOM, then $\mathcal{D} \subseteq \mathcal{C}$. 
\end{lemma}
\begin{proof}
Let $\mathcal{D}$ be such that $\mathcal{SP}\subseteq \mathcal{D} \subseteq \mathcal{U}$ and let $\varphi$ be an own-peak-only rule defined on $\mathcal{E}_\mathcal{D}$ that is efficient, NOM, and meets the equal division guarantee. Assume, by way of contradiction, that there is $R_0 \in \mathcal{D}\setminus \mathcal{C}$. Then, there are $x,y,z \in \mathbb{R}_+$ such that $x<y<z$, $xP_0y$ and $zP_0y$. Without loss of generality, we can assume that $zR_0x$ (the case $xR_0z$ is symmetric).

Let $x' =\max\{w\in [x,y]  \text{ such that } wI_0 x\}$ and $y' =\min\{w\in [y,z]  \text{ such that } wI_0 x\}$. As $R_0$ is continuous, both $x'$ and $y'$ are well defined, $x'<y<y'$, and $y'P_0w$ for each $w\in (x',y')$. Let $(R, \Omega) \in \mathcal{E}_{\mathcal{D}}$ be such that $\Omega =n y'$, $R_1=R_0$, 
and $R_{-1} \in \mathcal{SP}^{n-1}$ is such that $p(R_j)\in (y',\frac{\Omega-x'}{n-1})$ and $\Omega R_j y'$ for each $j\in N\setminus\{1\}$. By Lemma \ref{between}, $\varphi_j(R, \Omega)\in [y',p(R_j)]$ for each $j\in N\setminus\{1\}$. Then, by feasibility and the fact that $p(R_j)<\frac{\Omega-x'}{n-1}$, it follows that  $\varphi_1(R, \Omega)\in (x',y']$. If $\varphi_1(R, \Omega)\in (x',y')$, by the definitions of $x'$ and $y'$, $y'P_1\varphi_i(R, \Omega)$. If we consider $R'_1 \in \mathcal{SP}$ be such that $p(R'_1 ) =\frac{\Omega}{n}=y'$, by the equal division guarantee $R'_1$ is an obvious manipulation of $\varphi$ at  $(R_1,\Omega)$, contradicting that $\varphi$ satisfies NOM. Therefore, $\varphi_1(R,\Omega)=y'$. Hence, again by Lemma \ref{between} and feasibility,  $\varphi_j(R, \Omega)=y'$ for each $j\in N$. Now, consider the feasible allocation %$(x',\frac{\Omega-x'}{n-1},\ldots,\frac{\Omega-x'}{n-1})$ 
in which agent $1$ gets $x'$ and the remaining  agents get $\frac{\Omega-x'}{n-1}$. Since $x'I_1y'$ and $\frac{\Omega-x'}{n-1} P_j y'$ for each $j \in N\setminus\{1\}$, this new allocation Pareto improves upon $\varphi(R,\Omega),$ contradicting efficiency. We conclude that $\mathcal{D} \subseteq \mathcal{C}$. 
\end{proof}

%\begin{remark} \label{effsame}
%As we previously said efficiency is equivalent to same-sidedness on $\mathcal{E}_\mathcal{SP}$. On $\mathcal{E}_\mathcal{C}$ this equivalence no longer holds, but efficiency implies that: 
    
%(i) If $\Omega \leq \sum_{i\in N}\underline{p}(R_i),$ then $\varphi_i(R, \Omega)\leq \underline{p}(R_i)$ for all $i\in N$ 
    
%(ii) If $\sum_{i\in N}\underline{p}(R_i)\leq\Omega \leq \sum_{i\in N}\overline{p}(R_i),$ then $\varphi_i(R, \Omega)\in [\underline{p}(R_i), \overline{p}(R_i)]$ for all $i\in N$ 
    
%    (iii) 

%    (SEE Remark 1 in CHING AND SERIZAWA )
%\end{remark}

\begin{theorem}\label{theo max domain}
Domain $\mathcal{SPL}$ is a maximal domain for own-peak-onliness, efficiency, the equal division guarantee, and NOM. 
\end{theorem}
\begin{proof}
Let $\mathcal{D}$ be a domain for the properties listed in the theorem such that $\mathcal{SPL}\subseteq \mathcal{D}$. By Lemma \ref{convex}, $\mathcal{D}\subseteq \mathcal{C}$. Assume, by way of contradiction, that there is $R_0 \in \mathcal{D}\setminus \mathcal{SPL}$. Without loss of generality, we can assume that there are $x,y \in \mathbb{R}_+$ such that $x<y<\underline {p}(R_0)$ and $xI_0y$. 

Assume further that there is a rule $\varphi$ defined on $\mathcal{E}_{\mathcal{D}}$ that satisfies the properties listed in the Theorem. Consider $(R, \Omega) \in \mathcal{E}_{\mathcal{D}}$ with $\Omega =n y$, $R_{1}=R_{0}$, and $R_{-1} \in \mathcal{SP}^{n-1}$ such that $p(R_j)=y+\frac{y-x}{n-1}$  for each $j\in N\setminus\{1\}$. By Lemma \ref{between}, $\varphi_j(R,\Omega) \in [y,p(R_j)]$ for each  $j\in N\setminus\{1\}$. Then, by feasibility, $\varphi_1(R, \Omega) \leq y.$ If $\varphi_1(R, \Omega)<y$ we can assume, by own-peak-onliness and the fact that $\mathcal{SPL}\subseteq \mathcal{D}$, that $yP_1\varphi_1(R, \Omega)$. Next, let us consider $R'_1 \in \mathcal{SP}$ such that $p(R'_1 ) =\frac{\Omega}{n}=y.$ By the equal division guarantee, $R'_1$ is an obvious manipulation of $\varphi$ at $R_i$, contradicting the fact that $\varphi$ satisfies NOM. Therefore, $\varphi_1(R, \Omega)=y$. Hence, again by Lemma \ref{between} and feasibility,  $\varphi_j(R, \Omega)=y$ for each $j\in N$. Now, consider the feasible allocation %$(x',\frac{\Omega-x'}{n-1},\ldots,\frac{\Omega-x'}{n-1})$ 
in which agent $1$ gets $x$ and the remaining  agents get $y+\frac{y-x}{n-1}$. Since $xI_1y$ and $y+\frac{y-x}{n-1} P_j y$ for each $j \in N\setminus\{1\}$, this new allocation Pareto improves upon $\varphi(R,\Omega),$ contradicting efficiency. We conclude that $\mathcal{D}=\mathcal{SPL}$.

%The feasible allotment in which agent $1$ gets $x$, and each one of the other gets $y+\frac{y-x}{n-1}$ Pareto improves upon $\varphi(R),$ which contradicts EFF.  

The proof is completed by considering the extension of the uniform rule to  domain $\mathcal{SPL}$ defined in the proof of the Theorem in \cite{ching1998maximal}. As they observe, this rule satisfies efficiency and strategy-proofness on  $\mathcal{SPL}$. As NOM is a weakening of strategy-proofness, this rule satisfies NOM as well. Furthermore, it is easy to see that this rule is own-peak-only and meets the equal division guarantee.\footnote{We can also consider extensions of simple rules defined on $\mathcal{E}_{\mathcal{SP}}$ to domain $\mathcal{E}_{\mathcal{SPL}}$. Before defining them, we need some notation. Given $(R,\Omega) \in \mathcal{E}_{\mathcal{SPL}},$  let $\underline{z}(R, \Omega)=\sum_{j \in N}\underline{p}(R_j)-\Omega$ and $\overline{z}(R, \Omega)=\sum_{j \in N}\overline{p}(R_j)-\Omega.$ For $i \in N$ and $R_i \in\mathcal{SPL}$, let $\underline{R}_i \in \mathcal{SP}$ be such that $p(\underline{R}_i)=\underline{p}(R_i)$ and let $\overline{R}_i \in \mathcal{SP}$ be such that $p(\overline{R}_i)=\overline{p}(R_i)$. An own-peak-only rule $\varphi^\star$ defined on  $\mathcal{E}_{\mathcal{SPL}}$ is an \textbf{extension} of simple rule $\varphi$ defined on $\mathcal{E}_{\mathcal{SP}}$ if, for each $ (R, \Omega) \in \mathcal{E}_{\mathcal{SPL}},$
%Given a simple rule $\varphi: \mathcal{E}_\mathcal{SP} \longrightarrow X$ we can extend this rule to $\varphi^*: \mathcal{E}_\mathcal{SPL} \longrightarrow X$ as an own-peak-only rule that for each 
\begin{enumerate}[(i)]
    \item $\varphi^\star_i(R, \Omega)=\varphi_i(\underline{R}, \Omega)$ if $\underline{z}(R,\Omega) \geq 0$,
    \item $\varphi^\star_i(R, \Omega)=\varphi_i(\overline{R}, \Omega)$ if $\overline{z}(R,\Omega) \leq 0$,
    \item $\varphi^\star_i(R, \Omega)\in [\underline{p}(R_i), \overline{p}(R_i)]$ if $\underline{z}(R,\Omega)<0$ and $\overline{z}(R,\Omega)>0$. 
\end{enumerate}
%where $\varphi$ is an simple rule on $\mathcal{E}_{\mathcal{SP}}.$
With similar arguments to those used in the proof of Theorem \ref{characterization},  it follows that these rules also are efficient, meet the equal division guarantee, and are NOM. 
%on $\mathcal{E}_{\mathcal{SPL}}.$
}    
\end{proof}

Finally, %concerning the maximality result of Section \ref{section maximal}, 
it is worth mentioning that the single-plateaued domain remains maximal for the properties involved in the characterization of Corollary \ref{characterization 2}, i.e., when symmetry is also considered.

%\textbf{Same-sidedness:} For each $(R,\Omega) \in \mathcal{E}_\mathcal{D}$,  $\sum_{j \in N}\underline{p}(R_j) \geq \Omega$ implies  $\varphi_i(R)\leq \underline{p}(R_i)$ for each $i \in N,$ 
%and  $\sum_{j \in N}\overline{p}(R_j) \leq \Omega$ implies  $\varphi_i(R)\geq \overline{p}(R_i)$ for each $i \in N.$
%\vspace{5 pt}

%Given economy $(R, \Omega) \in \mathcal{E}_{\mathcal{D}},$ agent $i \in N$ is  \textbf{simple} if either $\underline{z}(R,\Omega) \geq 0$ and $\underline{p}{R_i)<\frac{\Omega}{n}$ or $\overline{z}(R, \Omega)\leq 0$ and $\overline{p}(R_i)>\frac{\Omega}{n}$. Let $N^\star(R, \Omega)$ denote the set of simple agents of economy $(R, \Omega)$.   

%\begin{definition}\label{def simple}
%An own-peak-only rule $\varphi: \mathcal{E}_\mathcal{SDL} \longrightarrow X$ is an \textbf{simple rule*} if for each $ (R, \Omega) \in \mathcal{E}_{\mathcal{SPL}}$ and each $i \in N,$
%\begin{enumerate}[(i)]
%    \item $\varphi_i(R, \Omega)$ is between $\frac{\Omega}{n}$ and $p(R_i)$ for each $i \in N$, if  
%    \item $\varphi_i(R, \Omega)=p(R_i)$ if  $i \in N^\star(R, \Omega)$.
    % \item $\varphi_i(R, \Omega)\leq \varphi_j(R, \Omega)$ if $p(R_i)\leq p(R_j)$.
%\end{enumerate} 
%\end{definition}

\section{Final Remarks}\label{section final}

This paper presents several families of efficient and own-peak-only allocation rules that satisfy the NOM  requirement. Table \ref{tabla caracterizaciones} summarizes the different characterizations obtained. Our main result (Theorem \ref{characterization}) shows that relaxing strategy-proofness to NOM generates an abundance of simple rules. This embarrassment of riches follows from the fact that for each conceivable claims rule, there is a corresponding simple allocation rule. Note, however, that such expansion can be performed without any essential change in the maximality of the domain of preferences involved, which continues to be single-plateaued.     

%Nuestra principal characterization (Teorema 1,2 ) muestra que cuando SP es relajada a NOM una enorme clase de reglas (simple rules) emergen. La amplitude de esta clase queda de manifiesto por el echo de que para cada claim rule definidad un problems de claim,  podemos identificar una simple rule. However, cuando nos concentramos en el dominion maximal no podemos ir mucho mas alla de single pleak, similarmente a lo que ocurre cuando strategy proof es considered. 

\begin{table}[h] 
\small
\centering 
\begin{threeparttable}
\begin{tabular}{|l|c|c|c|c|}
\hline
  & Th. \ref{characterization}  & Cor. \ref{characterization 2}  & Prop. %\ref{characterization peaks-only} & Th. 
  \ref{characterization bis} & Prop. \ref{characterization3} %& Prop. \ref{characterization generalized simple}
  \\
 \hline \hline
Own-peak-onliness & $+$ & $+$ & $+$ & $+$ \\
\hline 
Efficiency & $+$ & $+$ & $+$ & $+$ \\
\hline 

Equal division guarantee & $+$ & $+$ &  &    \\
\hline 
NOM & $+$& $+$& $+$& $+$ \\
\hline 
Symmetry & &$+$& & \\
\hline 
%Peaks-onliness & & &$+$& & \\
%\hline 
Endowments guarantee & & & $+$ &  \\
\hline 
Peak responsiveness &  & & & $+$  \\
\hline 

\end{tabular}
\end{threeparttable}
\caption{\emph{Characterizations of NOM rules.}}\label{tabla caracterizaciones}
\end{table}

%In the special case where we consider symmetry as a desideratum, imposing strategy-proofness also forces us to apply the uniform rule. However, if NOM is a sufficiently good requirement to describe the strategic behavior of the agents, 

In the specific scenario where symmetry is a desired property, enforcing strategy-proofness leads us to adopt the uniform rule. However, if NOM adequately captures the strategic behavior of agents,
then the wide array of symmetric simple rules emerges to consider (Corollary \ref{characterization 2}). Moreover, this family is versatile enough to, at least partially, accommodate other interesting principles in the allocation process. Specifically, we have seen that proportionality and the equal distance criterion can be made (partially) compatible with NOM.    

%Decir que SP con simetria y eficiencia reduce la clase a unica regla pero si consideramos NOM en vez de SP otras reglas interesante emenrgen y podemos considered otros criterios relevantesn en la literature a lo hora de repartir. En particular, proporcionalidad e equal distance pueden ser parcialmente compatibilizados con NOM.

%Decir algo de las Claims rules que se pueden generar una enorme clase de simple rules, pero que a pesar de eso el domino maximal asoiado a las propiedades dades consitnua siendo escencialemente el single plato.

%Preguntas. Un pregunta interesante sería como recuperar la unicidad de la uniforme cuando tenemos NOM en lugas de SP. 

An interesting problem for future research involves exploring NOM rules defined on the unrestricted domain of preferences. Nevertheless, pursuing this direction, as indicated by our Theorem \ref{theo max domain}, would necessitate sacrificing own-peak-onliness, efficiency, or the equal division guarantee.

\bibliographystyle{ecta}
\bibliography{biblio-obvious}

\appendix

\section{On the equal division guarantee}\label{on the EDG}

Remember that a rule is envy-free if no agent prefers some other agent's allotment to his own; and that a rule meets the equal division lower bound if, for each agent, the allotment recommended by the rule is at least as good as equal division. Formally,

%In this final section, we present five (VER) remarks that relate the properties 
%present some relations of the properties introduce in Subsection \ref{2.1} in relation to other properties known in the literature. Such relations leave interesting variants of Theorem \ref{characterization} and Corollary  \ref{characterization 2}.   

%\subsection{Envy-free}

%The second remark is with respect to EDG. By definition, it is clear that EDG is a minimal fairness condition. Furthermore, it is implied by the well-known envy-free condition for own-peak-only rules on $\mathcal{SP}$.  

\vspace{5 pt}
\noindent
\textbf{Envy-free:} For each   $(R,\Omega) \in \mathcal{E}_\mathcal{SP}$ and each $\{i,j\} \subseteq N$ with $i\neq j$,  $\varphi_i(R, \Omega)R_i\varphi_j(R, \Omega)$.
\vspace{5 pt}

\vspace{5 pt}
\noindent
\textbf{Equal division lower bound:\footnote{This property is also known as \emph{individual rationality from equal division.}}} For each   $(R,\Omega) \in \mathcal{E}_\mathcal{SP}$ and each $i \in N$,   $\varphi_i(R, \Omega)R_i \frac{\Omega}{n}.$
\vspace{5 pt}

\begin{lemma}
    For any own-peak-only rule defined on  $\mathcal{E}_\mathcal{SP}$,  envy-freeness or the equal division lower bound imply  the equal division guarantee.
\end{lemma}
\begin{proof}
    Let $\varphi$ be an own-peak-only rule defined on  $\mathcal{E}_\mathcal{SP}$. First, assume that $\varphi$ meets the equal division lower bound.  Let $(R,\Omega) \in \mathcal{E}_\mathcal{SP}$ and  $i \in N$ be such that $p(R_i)=\frac{\Omega}{n}$. Then, $\varphi_i(R,\Omega)R_i\frac{\Omega}{n}=p(R_i)$ and thus $\varphi_i(R, \Omega)=p(R_i)$, implying that $\varphi$ meets the equal division guarantee. Second, assume that $\varphi$ is envy-free but does not meet the equal division guarantee. Then, there are $(R,\Omega)\in \mathcal{E}_\mathcal{SP}$ and $i\in N$ with $p(R_i)=\frac{\Omega}{n}$ such that $\varphi_i(R)\neq\frac{\Omega}{n}$. Consider the case $\varphi_i(R)>\frac{\Omega}{n}$ (the other one is symmetric). 
    By feasibility, there is $j \in N\setminus \{i\}$ such that $\varphi_j(R)<\frac{\Omega}{n}$. As $p(R_i)=\frac{\Omega}{n}$, by own-peak-onliness, we can assume that $R_i$ is such that  $\varphi_j(R)P_i\varphi_i(R)$. Then, $\varphi$ does not satisfy envy-freeness.  This contradiction shows that $\varphi$ meets the equal division guarantee.  
\end{proof}

\section{%Appendix. 
Constructing a simple rule}\label{appendix}

Even though simple rules are easy to define and we can think of several different ways to construct them, in this appendix we provide a particular algorithmic approach.\footnote{Of course, this process will be nothing more than a way to resolve the system of inequalities posed by the definition of simple rule once an economy is fixed.}

%Ya hemos definifo y caracterido las simple rules. Si bien su definicion es muy sencilla y podriamos pensar en muchos caminos para construir una regla simple, en este apendice queremos, damos un proceso algoritmico particular para  construirlas. Claro esta que este proceso no sera otra cosa que un modo de resolver el sistema de inecuaciones que determine una regla simple una vez que fijamos una economia...

%A simple rule can be constructed sequentially by an algorithm.
%The next algorithm shows how we can construct a simple rule. 
Once we specify an economy $(R, \Omega) \in \mathcal{E}_{\mathcal{SP}}$, the procedure in Figure \ref{algorithm}  shows how we can construct the outcome of a simple rule in such an economy.    

\begin{figure}
\small{
\begin{center}
\begin{tabular}{l l}
\hline \hline
\multicolumn{2}{l}{\textbf{Algorithm:}}\vspace*{10 pt}\\
\textbf{Input} & An economy $(R, \Omega) \in \mathcal{E}_{\mathcal{SP}}$  \\

& \\
\textbf{Initialization} &% For each $i \in N$,
\\
& $\alpha_i:=\begin{cases}
    p(R_i) & \text{ if } i \in N^+(R,\Omega)\\
    \frac{\Omega}{n} & \text{ if } i \in N^-(R,\Omega) \\
\end{cases}$ \hfill \textcolor{blue}{\textsf{(initial allotment)}}\\

%\textsf{($\underline{\Omega}$<0 (>0) is neccesary in order to no simple agents recive more (left) than or equal to their tops) }
%\\
%& \\
\textbf{Adjustment} &  \\
& $N^-(R,\Omega):=\{i_1, i_2, \ldots, i_k\}$ \hfill \textcolor{blue}{\textsf{(order $N^-(R,\Omega)$)}}\\
&  $t:=1$\\

& $\overline{\Omega}_1:= \Omega- \sum_{j\in N} \alpha_j$ \hfill \textcolor{blue}{\textsf{(amount still left to allocate)}}\\

& $\underline{\Omega}_1:= \Omega- \sum_{
j\in N} p(R_j)$ \hfill \textcolor{blue}{\textsf{(%amount to to discount to non-simple agents from his top???
amount to control for type of case)}}\\

%&$\underline{\Omega}_1:= \overline{\Omega}_1- \sum_{j=1}^k p(R_{i_j})$ \hfill \textcolor{blue}{\textsf{ESTA CREO QUE ESTA MAL Y NO VA}}\\

\\

%\textbf{Case A}
& \texttt{IF}  $z(R,\Omega)\geq 0$ \hfill \textcolor{blue}{(\textsf{excess demand case})}\\

%\item[
%\textbf{Step $t$ ($1\leq t<k$)} 
& \hspace{20 pt} \texttt{WHILE} $t < k$ \texttt{DO}\\

& \hspace{20 pt} Choose $\lambda_t$ such that \\
& \hspace{60 pt}$\max \{0,\underline{\Omega}_t+p(R_{i_t})-\frac{\Omega}{n}\}\leq\lambda_t\leq \min\{p(R_{i_t})-\frac{\Omega}{n},\overline{\Omega}_t\}$ \\
& \hspace{20 pt} $\alpha_{i_t}=\alpha_{i_t}+\lambda_t$ \\
& \hspace{20 pt} $\overline{\Omega}_{t+1}=\overline{\Omega}_t-\lambda_t$ \\
& \hspace{20 pt} $\underline{\Omega}_{t+1}=\underline{\Omega}_t-\alpha_{i_t}+p(R_{i_t})$ \\
& \hspace{20 pt} $t=t+1$\\

& \texttt{ELSE} \hfill \textcolor{blue}{(\textsf{excess supply case:}  $z(R,\Omega)< 0$)}\\
& \hspace{20 pt} \texttt{WHILE} $t < k$ \texttt{DO}\\

& \hspace{20 pt} Choose $\lambda_t$ such that \\
&  \hspace{60 pt}$\max \{0,\frac{\Omega}{n}-p(R_{i_t})-\underline{\Omega}_t\}\leq\lambda_t\leq \min \{\frac{\Omega}{n}-p(R_{i_t}), -\overline{\Omega}_t\}$ \\

& \hspace{20 pt} $\alpha_{i_t}=\alpha_{i_t}-\lambda_t$ \\
%& $\overline{\Omega}:\overline{\Omega}+(\lambda_i-\frac{\Omega}{n})$ \\
& \hspace{20 pt} $\underline{\Omega}_{t+1}=\underline{\Omega}_{t}-\alpha_{i_t}+p(R_{i_t})$ \\
& \hspace{20 pt} $t=t+1$\\

\textbf{Output} & %For each $i \in N$,
\\
& $\varphi_i(R,\Omega)=\begin{cases}
    \alpha_i & \text{ if } i \in N\setminus \{i_k\}\\
    \ \Omega-\sum_{j\in N\setminus \{i_k\}} \alpha_j & \text{ if } i=i_k\\
\end{cases}$ \hfill \textcolor{blue}{\textsf{(final allotment)}}\\
\\
\hline \hline
\end{tabular}
\end{center}
\caption{\emph{Algorithm to construct a simple rule.}}
\label{algorithm}
}
\end{figure}

%\bigskip
%We give a brief description of the algorithm. 
Assume that $z(R,\Omega)\geq 0$ (the other case is similar). %In the initialization, 
To begin, an initial allotment is performed. Each simple agent is assigned his peak amount and each non-simple agent is assigned equal division. This frees up an amount  $\overline{\Omega}_1\geq 0$ that still has to be divided among non-simple agents. To do this, we consider a sequential adjustment process. First, we order non-simple agents arbitrarily.  Following this order, in step $t$ of the adjustment process non-simple agent $i_t$'s  surplus,  %allotment is adjusted upwards to obtain a new allotment 
$\lambda_t$, is determined fulfilling three requirements: (i) agent $i_t$'s final allotment $\frac{\Omega}{n}+ \lambda_{t}$ should be between $\frac{\Omega}{n}$ and $p(R_{i_t})$ to comply with the definition of simple rule; (ii) surplus $\lambda_{t}$ should be small enough to achieve feasibility, i.e.,  $\lambda_{t} \leq \overline{\Omega}_t$; and (iii) surplus $\lambda_{t}$ should be large enough to maintain the economy in excess demand for not yet adjusted %the remaining 
agents, i.e., ${\Omega}-[\sum_{j\in N^+(R,\Omega) \cup \{i_1, \cdots, i_{t-1}\}}\alpha_j+(\frac{\Omega}{n}+\lambda_{t})]\leq \sum_{j=t+1}^k p(R_{i_j})$ or, equivalently, $\underline{\Omega}_t+p(R_{i_t})-\frac{\Omega}{n}\leq\lambda_t$.
%$\overline{\Omega}_t-(\frac{\Omega}{n}+\lambda_{i_t})\leq \sum_{j=t}^k p(R_{i_j})$. 
To finish step $t$, amounts $\overline{\Omega}_t$ and $\underline{\Omega}_t$ are updated accordingly. Finally, each simple agent gets as final allotment his peak, each non-simple agent  $i_t$ different from $i_k$ gets as final allotment $\frac{\Omega}{n}+\lambda_t$ and the allotment of agent $i_k$ is determined by feasibility. 

Notice that the algorithm is correct. To see this, observe that in each step $t$ of the algorithm $\alpha_i \leq p(R_i)$ for all $i\in N\setminus\{i_t\}$ and at the beginning of each step $t$, $\alpha_{i_t}=\frac{\Omega}{n} \leq p(R_{i_t})$ . Therefore, $\underline{\Omega}_t+p(R_{i_t})-\frac{\Omega}{n}\leq \overline{\Omega}_t$. Furthermore, $\underline{\Omega}_1\leq0$ in the initialization and $\underline{\Omega}_{2}=\underline{\Omega}_{1} -\alpha_{i_1}+p(R_{i_1})=\underline{\Omega}_{1} -(\frac{\Omega}{n}+\lambda_{1})+p(R_{i_1}).$ Therefore, by the choice of $\lambda_{1}$,  $\underline{\Omega}_2\leq0$. Following a similar argument we get that $\underline{\Omega}_t\leq0$ at each step of the algorithm. Therefore, $\underline{\Omega}_t+p(R_{i_t})-\frac{\Omega}{n}\leq p(R_{i_t})-\frac{\Omega}{n}$ holds. Then, the election of $\lambda_t$ is always possible.

\end{document}